\documentclass[12pt,a4paper]{article}

\usepackage[margin=1in]{geometry} %
\usepackage[utf8]{inputenc} %
\usepackage{authblk} %
\usepackage{hyperref} %
\usepackage[numbers,sort]{natbib} %
\usepackage{color}
\usepackage{array,multirow}

\usepackage{enumitem}
\setlist{itemsep=0pt}

\usepackage{amssymb,amsmath,amsthm}
\usepackage{mathtools}
\usepackage{graphicx}
\usepackage{subcaption}
\usepackage[ruled]{algorithm2e}
\usepackage[capitalise]{cleveref}

\theoremstyle{definition}
\newtheorem{theorem}{Theorem}[section]
\newtheorem{proposition}{Proposition}
\newtheorem{lemma}{Lemma}
\newtheorem{corollary}{Corollary}

\theoremstyle{remark}
\newtheorem{remark}[theorem]{Remark}

\newcommand{\R}{\mathbb{R}}
\newcommand{\E}{\mathbb{E}}

\DeclareMathOperator{\trace}{trace}

\DeclareMathOperator{\Jac}{Jac}

\DeclareMathOperator{\Cov}{Cov}
\DeclareMathOperator{\Dkl}{D_\mathrm{KL}}
\DeclareMathOperator*{\argmin}{arg\,min}
\DeclareMathOperator*{\argmax}{arg\,max}

\renewcommand{\d}{\mathrm{d}}

\usepackage{xargs}                      %
\usepackage[pdftex,dvipsnames]{xcolor}  %
\usepackage[colorinlistoftodos,prependcaption,textsize=tiny,textwidth=0.85in]{todonotes}
\newcommandx{\oz}[2][1=]{\todo[linecolor=OliveGreen,backgroundcolor=OliveGreen!25,bordercolor=OliveGreen,#1]{#2}}

\usepackage[normalem]{ulem} %

\begin{document}

\title{
Certified Bayesian optimal experimental design: from Expected Information Gain to Signal-to-Noise Ratio
}

\author{
Mohamed Doumbouya\footnote{mohamed.doumbouya@inria.fr},
Arthur Vidard\footnote{arthur.vidard@inria.fr},
Olivier Zahm\footnote{olivier.zahm@inria.fr}
}

\maketitle

\begin{abstract}
 Optimal experimental design is often formulated as the maximization of the Expected Information Gain (EIG), which measures the expected reduction in uncertainty about a parameter of interest after observing the data. For nonlinear forward models and nonGaussian priors, computing the EIG typically requires costly nested Monte Carlo estimators or sophisticated density-approximation techniques. In this work, we derive new upper and lower bounds on the EIG using dimensional logarithmic Sobolev inequalities. The proposed bounds are expressed in terms of signal-to-noise ratios (SNR), preserving the simple matrix-based structure of the EIG in the linear-Gaussian setting and depending solely on covariance matrices and Fisher information matrices. We introduce two families of bounds: conservative bounds, designed to retain the entire information content, and incremental bounds, that prioritize early information acquisition for sequential experimental design. When model gradients are available, we further derive computable approximations of the Fisher-information terms that require only a small number of gradient evaluations, yielding substantial computational savings.
 The tightness of these bounds is controlled by the nonlinearity of the forward model, and in the linear-Gaussian setting, these bounds recover the exact value of the EIG.  By eliminating the nested-sampling bottleneck, the proposed framework offers a tractable and theoretically grounded alternative to direct EIG estimation for large-scale Bayesian experimental design, while retaining the interpretability of the SNR-based bounds.
\end{abstract}


\section{Introduction}

Optimal experimental design (OED) aims to select a limited number of potentially costly experiments that maximize the information which will be gained about a parameter of interest.
Among the many criteria used to quantify the informativeness of a design, the Expected Information Gain (EIG) has become one of the most widely used in Bayesian statistics.
It is defined as the expected Kullback-Leibler divergence between the prior and posterior distributions \cite{lindley1956measure,chaloner1995bayesian,ryan2016review}.
Formally, denoting by $\theta$ the \emph{parameters of interest} with prior density $\pi_\theta$, and by $Y$ the random vector representing the \emph{observable data}, the problem is
\begin{equation}\label{eq:EIG_intro}
 \max_{W_m}~ \mathrm{EIG}(W_m) := \E\left[\Dkl( \pi_{\theta|Y_m} || \pi_\theta ) \right] ,
\end{equation}
where $\pi_{\theta\mid Y_m}$ denotes the posterior density of $\theta$ given the observations $Y_m = W_m^\top Y$. Here, $Y_m$ represents the \emph{observed data} extracted from $Y$ under the design $W_m$, modeled as a matrix with $m$ columns. Hence, maximizing the EIG selects the design $W_m$ that induces, on average over $Y$, the largest prior-to-posterior update. We refer to \cite{huan2024optimal} for a more in-depth presentation.

Computing $\mathrm{EIG}(W_m)$ remains challenging in general. Standard estimators rely on nested Monte Carlo sampling over both the data distribution and the resulting posterior distributions, leading to high computational costs and slow convergence \cite{ryan2003estimating,huan2013simulation}. Existing approaches either accelerate the nested estimation, for example through Laplace-based importance sampling \cite{long2013fast,beck2018fast} or avoid it by approximating the relevant probability densities directly, using variational methods \cite{foster2019variational}, neural mutual-information estimators \cite{kleinegesse2020bayesian}, or transport-map surrogates \cite{li2024expected,koval2024tractable}. We refer to \cite{rainforth2024modern} for a recent review of these computational aspects. Although these methods significantly broaden the applicability of EIG-based design, they often remain computationally demanding and provide limited theoretical guarantees.
However, the situation is remarkably simpler in the linear-Gaussian setting.
There, the EIG admits the closed-form expression $\mathrm{EIG}(W_m)  = \frac{1}{2}  \ln ( \frac{|  \Cov(\theta)   | }{ |  \Cov(\theta|Y_m) |})$, where the posterior covariance $\Cov(\theta|Y_m)$  is independent of the observed data $Y_m$ but depends on the design matrix $W_m$.
This expression is at the heart of the widely used D-optimality criterion \cite{lindley1956measure,chaloner1995bayesian}.
Interestingly, the EIG of linear-Gaussian problems can also be expressed in terms of a signal-to-noise ratio (SNR) as
\begin{equation}\label{eq:EIG_linearGaussian_intro}
 \mathrm{EIG}(W_m)
 = \frac{1}{2}  \ln \left( \frac{|  W_m^\top \Cov(Y) W_m    | }{ |  W_m^\top \Cov(Y|\theta) W_m |}\right),
\end{equation}
where $\Cov(Y)$ is the covariance of the observable data (the \emph{signal}) and $\Cov(Y|\theta)$ is the covariance of $Y$ after freezing $\theta$ (the \emph{observation noise}).

In this work, we revisit the computation of the EIG for nonlinear forward models and non-Gaussian priors, following the perspective of \cite{chen2025coupled}, who cast Bayesian OED as a dimension reduction problem of the observable data.
Functional inequalities have recently become a popular tool for the gradient-based dimension reduction in uncertainty quantification tasks, for instance in function approximation \cite{constantine2014active,verdiere2025diffeomorphism,zahm2020gradient,verdiere2025mollified}, density approximation \cite{zahm2022certified,flock2024certified,li2025sharp,li2024principal}, rare event estimation \cite{uribe2021cross,breaz2026new,jiang2021recursive}
or sensitivity analysis \cite{kucherenko2016derivative,song2019derivative,heredia2026weighted,heredia2025one,roustant2017poincare}.
In information theory, the differential entropy is also often bounded by the gradient-based Fisher information matrices via \emph{e.g.} Cramér-Rao or Efroimovich's inequality \cite{aras2019family,efroimovich1980information}.
Here, by leveraging dimensional logarithmic Sobolev inequalities \cite{bakry2012dimension}, we derive new upper and lower bounds for the EIG that retain the same simple algebraic structure as in \eqref{eq:EIG_linearGaussian_intro}.
We obtain two families of SNR bounds. The first, which we refer to as \emph{conservative bounds}, quantify how much information from the entire observable dataset is preserved by a candidate design.
Optimizing the conservative bounds corresponds to solving
\begin{equation}\label{eq:optimization_con_intro}
 \max_{W_m}  ~ \frac{1}{2} \ln \left( \frac{|  W_m^\top \Cov(Y) W_m    | }{ |  W_m^\top \Sigma_\mathrm{noise} W_m |}\right),
\end{equation}
where $\Sigma_\mathrm{noise}$ relates to the Fisher information matrix of $Y|\theta$.
We note that the conservative lower-bound on the EIG was recently established in \cite{cui2025subspace} using different proof techniques.
The second family, which we call the \emph{incremental bounds}, are tailored to sequential design and quantify the additional information provided by a new observation beyond an existing dataset.
Optimizing these bounds yields
\begin{equation}\label{eq:optimization_inc_intro}
 \max_{W_m}  ~ \frac{1}{2} \ln \left( \frac{|  W_m^\top \Sigma_\mathrm{signal} W_m    | }{ |  W_m^\top \E[\Cov(Y|\theta)] W_m |}\right)
\end{equation}
where $\Sigma_\mathrm{signal}$ relates to the Fisher information matrix of $Y$.
We prove that the solutions of \eqref{eq:optimization_con_intro} and \eqref{eq:optimization_inc_intro} are quasi-optimal solutions to problem \eqref{eq:EIG_intro}, with quasi-optimality constants that are explicitly controlled by some measure of the nonlinearity and nonGaussianity of the underlying data-generating process. However, we emphasize that the solutions of \eqref{eq:optimization_con_intro} and \eqref{eq:optimization_inc_intro} should not only be viewed as approximations to \eqref{eq:EIG_intro}: they can be adopted as novel BOED criteria in their own, offering substantially more favorable computational strategies than the original criterion \eqref{eq:EIG_intro}.

Both families of bounds require computing the four covariance/Fisher matrices appearing in \eqref{eq:optimization_con_intro}--\eqref{eq:optimization_inc_intro}. We discuss several ways for estimating these quantities, depending on the structure of the relationship between $\theta$ and $Y$. We first consider a sampling-based approach, which relies on pick-and-freeze estimators that are widely used in global sensitivity analysis, see \cite{gamboa2016statistical,janon2014asymptotic}. Second, we consider an approximation-based approach that relies on estimating conditional expectations. This strategy is related in spirit to the approach developed in \cite{hoang2025scalable}, but differs significantly in the quantity being approximated: whereas \cite{hoang2025scalable} estimates the conditional expectation of the parameter of interest given the \emph{observed} data $Y_m$, our approach estimates the conditional expectation of the \emph{forward model} given the \emph{observable} data $Y$ and the parameter of interest $\theta$. Finally, when gradients of the forward model are available, we derive additional bounds on the Fisher-information-type matrices that can be estimated with a small number of gradient evaluations of the forward model. This leads to substantially lower computational costs at the expense of looser EIG bounds. We show that the resulting gap is controlled by the nonlinearity of the forward model and vanishes in the linear-Gaussian limit, where the exact EIG is recovered.

The rest of the paper is organized as follows. In Section \ref{sec:theory}, we present our main theoretical results, which is to derive upper and lower bounds for the EIG (Theorem \ref{th:Bound_EIG}) and then the incremental and conservative bounds (Corollaries \ref{cor:Bound_EIG_incremental} and \ref{cor:Bound_EIG_conservative}). In Section \ref{sec:ComputeDiagnostic}, we discuss computational strategies for estimating the four diagnostic matrices appearing in these bounds. Finally, in Section \ref{sec:numerics}, we present numerical experiments illustrating the proposed methodology on a parametrized PDE.

\section{Problem statement and main results}\label{sec:theory}

Let $\theta$ be a $\R^d$-valued random vector which represents a quantity of interest, and let $Y$ be a $\R^p$-valued observable random vector. We seek to infer $\theta$ from a partial observation of $Y$ defined by
$$Y_m= W_m^\top Y  , $$
for some (deterministic) design matrix $W_m\in\R^{p\times m}$ with $m\leq p$ columns.
Our objective is to identify $W_m$ which allows us to best infer $\theta$ from $Y_m$.
From the Bayesian perspective, the \emph{expected information gain} (EIG) measures the quantity of information on $\theta$ contained in $Y_m$, and is classically defined as the expected Kullback-Leibler (KL) divergence from the prior density $\pi_\theta$ of $\theta$ to the posterior density $\pi_{\theta|Y_m}$ of $\theta$ conditioned on $Y_m$. Formally, we aim to solve
\begin{equation}\label{eq:EIG}
 \max_{W_m\in \mathcal{K}_m\subseteq\, \R^{p\times m}} \mathrm{EIG}(W_m) := \E_{Y_m} \left[\Dkl( \pi_{\theta|Y_m} || \pi_\theta ) \right] ,
\end{equation}
where $\Dkl(\rho||\pi)=\int \log(\rho/\pi)\d\rho$ denotes the KL divergence and where $\mathcal{K}_m$ denotes the set of admissible design matrices.
For instance, when $\mathcal{K}_m = \{ W_m=[e_{\tau_1},\hdots,e_{\tau_m}] : \tau\subset\{1,\hdots,p\}, \#\tau=m \}$ is the set of all matrices whose columns are canonical vectors $e_1,\hdots,e_d$ of $\R^d$, the partial dataset $Y_m = (Y_{\tau_1},\hdots, Y_{\tau_m})$ is the vectors containing the components of $Y$ indexed by $\tau$.

It is well known that $\mathrm{EIG}(W_m)$ is equal to the \emph{mutual information} $\mathrm{EIG}(W_m) = \Dkl( \pi_{\theta Y_m} || \pi_{\theta} \pi_{Y_m} ) $ between $\theta $ and $Y_m$.
Thus, Problem \eqref{eq:EIG} is equivalent to finding the design matrix $W_m$ which maximizes the information that $Y_m$ carries about $\theta$. Similarly, we can write
\begin{align}
 \mathrm{EIG}(W_m)
 &= \int  \log\frac{\pi_{\theta|Y}}{\pi_{\theta}} \d\pi_{\theta Y} - \int  \log\frac{\pi_{\theta|Y}}{\pi_{\theta|Y_m}} \d\pi_{\theta Y} \nonumber\\
 &= \underbrace{\E_{Y} \left[\Dkl( \pi_{\theta|Y} || \pi_{\theta} ) \right]}_{=\mathrm{EIG}(I_p)} - \E_{Y} \left[\Dkl( \pi_{\theta|Y} || \pi_{\theta|Y_m} ) \right] ,  \label{eq:EIG_ErrorPost}
\end{align}
 so that Problem \eqref{eq:EIG} is also equivalent to minimizing $W_m\mapsto \E_{Y} \left[\Dkl( \pi_{\theta|Y} || \pi_{\theta|Y_m} ) \right] $ which measures how close the posterior $\pi_{\theta|Y_m}$ with partially observed data set $Y_m$ is to the reference posterior $\pi_{\theta|Y}$ in which the entire data set $Y$ is being observed.

The direct computation of $\mathrm{EIG}(W_m)$ is a difficult task in general as it involves high dimensional integrals which are not always tractable analytically nor computationally.
In this paper, we propose an affordable strategy, emphasising on the following prototypical settings of increasing complexity.

\begin{itemize}
  \item \textbf{Linear \& Gaussian setting.} The prior for $\theta$
  is Gaussian $\pi_\theta=\mathcal{N}(m_\theta,\Sigma_\theta )$ and the observable data is defined as
   $$
    Y = A\theta + \varepsilon ,
   $$
  where $A\in\R^{p\times d}$ is a (deterministic) matrix and $\varepsilon\sim\mathcal{N}(0,\Sigma_\mathrm{obs})$ a Gaussian noise which is independent of $\theta$. In this setting, the EIG can be analytically computed
 \begin{equation}\label{eq:EIG_Gaussian}
  \mathrm{EIG}(W_m)
  = \frac{1}{2}\ln \left( \frac{|  W_m^\top \Sigma_Y W_m    | }{ |  W_m^\top \Sigma_\mathrm{obs}  W_m |}\right) ,
 \end{equation}
  where $\Sigma_Y =  \Sigma_\mathrm{obs} + A \Sigma_\theta A^\top   $ is the covariance matrix of the observable data $Y$.
  We give the proof of this result in \ref{proof:EIG_Gaussian}. Intuitively, maximizing $\mathrm{EIG}(W_m)$ corresponds to finding the data components whose variance is largest relatively to the noise variance, hence the one that \emph{maximizes the Signal-to-Noise Ratio} (SNR).

  \item \textbf{Standard setting}. Given a prior density $\pi_\theta$ for $\theta$ and a forward model $G:\R^d\rightarrow\R^m$, the observable data set $Y\in\R^m$ is defined as
  \begin{equation}\label{eq:ForwardModel_nonGO}
  Y=G(\theta) + \varepsilon ,
  \end{equation}
  where $\varepsilon\sim\mathcal{N}(0,\Sigma_\mathrm{obs})$ is a Gaussian noise that is independent of $\theta$.
  This defines the likelihood function
  $
  \pi_{Y|\theta}(y|\theta)\propto \exp ( -\frac{1}{2}\| G(\theta) - y\|_{\Sigma_\mathrm{obs}^{-1}}^2 )
  $, where $\|\cdot\|_\Sigma^2=(\cdot)^\top \Sigma(\cdot)$.
  Without further assumptions on $G$ and on $\pi_\theta$, there is no closed form expression for $\mathrm{EIG}(W_m)$.

  \item \textbf{Auxiliary parameter setting.} Consider the following data generating process
  \begin{equation}\label{eq:AuxiliarySetting}
   Y = G(\theta,\eta) + \varepsilon
  \end{equation}
  where $\theta\sim\pi_\theta$ and $\eta\sim\pi_\eta$ and $ \varepsilon\sim\mathcal{N}(0,\Sigma_\mathrm{obs})$ are assumed to be independent. Here, $\eta$ is referred to as a \emph{auxiliary parameter} or a \emph{nuisance parameter}:
  the goal is to identify a design $W_m$ such that the resulting observations are primarily informative about the parameter of interest $\theta$, regardless of the auxiliary parameter $\eta$.

  \item \textbf{Goal-oriented setting}.
  In many situations, the observable data $Y$ is not a direct function of the parameter of interest $\theta$ as in  \eqref{eq:ForwardModel_nonGO} or \eqref{eq:AuxiliarySetting}, but the dependence of $Y$ on $\theta$ is rather described via a $\R^q$-valued \emph{latent variable} $\eta\sim \pi_\eta$ as in
  \begin{equation}\label{eq:GOsetting}
    \begin{split}
   Y&=G(\eta) + \varepsilon ,  \\
   \theta &= H(\eta)+\xi ,
  \end{split}.
 \end{equation}
  For simplicity, we assume that $\theta\sim\pi_\theta$ and $\eta\sim\pi_\eta$ and $\varepsilon\sim\mathcal{N}(0,\Sigma_\mathrm{obs})$ and $\xi\sim\mathcal{N}(0,\Sigma_\xi)$ are independent. Here, $G:\R^q\rightarrow\R^p$ denotes the forward model and $H:\R^q\rightarrow\R^d$ is a function which extracts the quantity of interest $\theta$, such as the energy associated with a state $\eta$ or any other physically meaningful quantity.
  In this setting, the objective is to determine the data $Y_m$ that are informative about $\theta$ primarily, hence the term \emph{goal-oriented}.
  Note that here, the prior $\pi_\theta$ is implicitly defined as $\pi_\theta(\theta)\propto\int \exp(-\frac{1}{2}\| H(\eta) - \theta\|_{\Sigma_\xi^{-1}}^2)\pi_{\eta}(\eta)\d\eta$, and the likelihood by
  \begin{equation*}\label{eq:GO_likelihood}
   \pi_{Y|\theta}(y|\theta) \propto  \frac{\int \exp(-\frac{1}{2}\| G(\eta) - y\|_{\Sigma_\mathrm{obs}^{-1}}^2) \exp(-\frac{1}{2}\| H(\eta) - \theta\|_{\Sigma_\xi^{-1}}^2 ) \pi_{\eta}(\eta) \d \eta }{ \int \exp(-\frac{1}{2}\| H(\eta) - \theta\|_{\Sigma_\xi^{-1}}^2)\pi_{\eta}(\eta)\d\eta } .
  \end{equation*}
  Since evaluating the prior and the likelihood is prohibitive (because of the integrals over $\eta$), this setting is also often referred to as the \emph{likelihood-free} setting, or as to the \emph{simulation-based} setting: one can easily draw samples of $(Y,\theta)$ via \eqref{eq:GOsetting}, but neither the likelihood $\pi_{Y|\theta}$ nor the prior $\pi_{\theta}$ can be evaluated in a straightforward manner.

\end{itemize}

It is worth noting that the Standard setting and the Auxiliary parameter settings are two particular cases of the Goal-oriented setting. Indeed, letting $H(\eta)=\eta$ and $\xi=0$ in \eqref{eq:GOsetting} yields $\theta=\eta$ and recovers $Y=G(\theta)+\varepsilon$ which is \eqref{eq:ForwardModel_nonGO}. Also, taking $\eta=(\eta_1,\eta_2)$, $H(\eta)=\eta_1$ and $\xi=0$ in \eqref{eq:GOsetting} yields $\theta = \eta_1$ and $Y=G(\theta,\eta_2)+\varepsilon$, which corresponds to \eqref{eq:AuxiliarySetting}.
Therefore, in what follows, we mainly focus on the goal-oriented setting \eqref{eq:GOsetting}.

We give now our main result. The proof relies on dimensional logarithmic Sobolev inequalities, see  \cite{bakry2012dimension} and Section \ref{sec:dLSI}, and involves \emph{Fisher information matrices}, see Section \ref{sec:Fisher} for more details.

\begin{theorem}[Main result]\label{th:Bound_EIG}
 Let $(\theta,Y)$ be a random vector on $\R^d\times \R^p$ with smooth joint probability density $\pi_{\theta,Y}$ such that $\E[\|Y\|^2]<\infty$.
 Let $\Sigma_Y=\Cov(Y)$ and $\Sigma_{Y|\theta}=\E[\Cov(Y|\theta)]$ and let $\Sigma_\mathrm{noise}\in\R^{p\times p}$ and $\Sigma_\mathrm{signal}\in\R^{p\times p}$ be any symmetric definite positive matrices that bound the (expected) Fisher information matrices as
 \begin{align}
   \Sigma_\mathrm{noise}^{-1} &\succeq \E[\mathcal{I}_{Y|\theta}] := \E_{(\theta,Y)}\left[ \left(\nabla_{Y}\ln \pi_{Y|\theta}(Y|\theta)   \right)\left(\nabla_{Y}\ln \pi_{Y|\theta}(Y|\theta)   \right)^\top  \right]  \label{eq:SigmaNoise} \\
   \Sigma_\mathrm{signal}^{-1}  &\succeq \mathcal{I}_{Y} := \E_Y\left[ \left(\nabla_{Y}\ln \pi_{Y}(Y)   \right)\left(\nabla_{Y}\ln \pi_{Y}(Y)   \right)^\top \right] . \label{eq:SigmaSignal}
 \end{align}
 Then, for any matrices $W_m\in\R^{p\times m}$ and $W_\mathrm{new}\in\R^{p\times m'}$ such that $[W_m,W_\mathrm{new}]\in\R^{p\times (m+m')}$ has rank $m+m'\leq p$, we have
  \begin{align}
  \mathrm{EIG}([W_m,W_\mathrm{new}])
 &\geq \mathrm{EIG}[W_m] + \frac{1}{2}    \ln\left( \frac{|W_\mathrm{new}^\top \Sigma_\mathrm{signal}[W_m]  W_\mathrm{new}|}{| W_\mathrm{new}^\top \Sigma_{Y|\theta}[W_m]  W_\mathrm{new} |}  \right)  \label{eq:Bound_EIG} \\
 \mathrm{EIG}([W_m,W_\mathrm{new}]) &\leq \mathrm{EIG}[W_m] + \frac{1}{2}   \ln\left( \frac{ | W_\mathrm{new}^\top \Sigma_{Y}[W_m] W_\mathrm{new} | }{ |W_\mathrm{new}^\top \Sigma_\mathrm{noise}[W_m] W_\mathrm{new}| }\right) , \label{eq:Bound_EIG_up}
 \end{align}
 where $[W_m,W_\mathrm{new}]$ denotes the concatenation of the columns of the matrices $W_m$ and $W_\mathrm{new}$ and
 where $\Sigma[W] = \Sigma - \Sigma W (W^\top \Sigma W)^{-1}W^\top \Sigma$.

\end{theorem}

\begin{proof}
 See Section \ref{proof:Bound_EIG_incremental}.
\end{proof}

The inequalities \eqref{eq:Bound_EIG}--\eqref{eq:Bound_EIG_up} involve two matrices $\Sigma_\mathrm{noise}$ and $\Sigma_\mathrm{signal}$ which are defined as bounds on the (expected) Fisher information matrices $\E[\mathcal{I}_{Y|\theta}]$ and $\mathcal{I}_Y$.
Although choosing $\Sigma_\mathrm{noise}=\E[\mathcal{I}_{Y|\theta}]^{-1}$ and $\Sigma_\mathrm{signal}=\mathcal{I}_{Y}^{-1}$ yields the sharpest bounds in \eqref{eq:Bound_EIG}--\eqref{eq:Bound_EIG_up}, the numerical computation of $\E[\mathcal{I}_{Y|\theta}]$ and $\mathcal{I}_Y$ may be prohibitively expensive depending on the context. In Section \ref{sec:ComputeDiagnostic} we propose bounds $\Sigma_\mathrm{noise}^{-1}\succeq\E[\mathcal{I}_{Y|\theta}]$ and $\Sigma_\mathrm{signal}^{-1}\succeq\mathcal{I}_{Y}$ which can be computed with significantly lower numerical effort.

\begin{corollary}[Incremental bounds]\label{cor:Bound_EIG_incremental}
Under the same assumptions as in Theorem \ref{cor:Bound_EIG_conservative}, for any rank-$m$ matrix $W_m\in\R^{p\times m}$ we have
  \begin{align}
  \mathrm{EIG}(W_m)
 &\geq   \frac{1}{2}    \ln\left( \frac{|W_m^\top \Sigma_\mathrm{signal}  W_m|}{| W_m^\top \Sigma_{Y|\theta}  W_m |}  \right)  \label{eq:Bound_EIG_incremental} \\
 \mathrm{EIG}(W_m) &\leq  \frac{1}{2}   \ln\left( \frac{ | W_m^\top \Sigma_{Y} W_m | }{ |W_m^\top \Sigma_\mathrm{noise} W_m| }\right) , \label{eq:Bound_EIG_incremental_up}
 \end{align}
\end{corollary}
\begin{proof}
 Apply Theorem \ref{th:Bound_EIG} with $W_m\leftarrow \emptyset$ and $W_\mathrm{new}\leftarrow W_m$.
\end{proof}

The bounds \eqref{eq:Bound_EIG_incremental}--\eqref{eq:Bound_EIG_incremental_up} are similar to the analytical expression \eqref{eq:EIG_Gaussian} of the EIG in the linear-Gaussian setting, in that the dependence on the design matrix $W_\mathrm{new}$ arises through a generalized Rayleigh quotient $|W_m^\top A W_m| / |W_m^\top B W_m| $.
We note that the bound \eqref{eq:Bound_EIG_incremental_up} has been recently derived in \cite[Section 3.1.3]{cui2025subspace} in the standard setting \eqref{eq:ForwardModel_nonGO} using other proof techniques.
Fundamentally, these bounds quantify the \emph{incremental information gain} obtained by observing $Y_m=W_m^\top Y$ (meaning $\mathrm{EIG}(W_m)$) compared to observing nothing (meaning $\mathrm{EIG}(\emptyset)=0$).

In contrast, the following corollary establishes \emph{conservative bounds} for the EIG.
These bounds compare the EIG achieved with $Y_m=W_m^\top Y$ to the total EIG that is obtainable by observing the whole dataset $Y$ (meaning $\mathrm{EIG}(I_d)$). When using these bounds to select $W_m$, the approach prioritizes retaining as much information about $\theta$ from $Y$ as possible, hence the term \emph{conservative}.

\begin{corollary}[Conservative bounds]\label{cor:Bound_EIG_conservative}
 Under the same assumptions as in Theorem \ref{cor:Bound_EIG_conservative}, we have
 \begin{align}
  \mathrm{EIG}(W_m)
  &\geq \mathrm{EIG}(I_p) + \frac{1}{2} \ln \left(\frac{| W_m^\top \Sigma_Y W_m |}{| W_m^\top \Sigma_\mathrm{noise} W_m |}\right)   -  \frac{1}{2} \ln\left(\frac{| \Sigma_Y |}{|\Sigma_\mathrm{noise}| }\right)   \label{eq:Bound_EIG_conservative}   \\
  \mathrm{EIG}(W_m)  &\leq \mathrm{EIG}(I_p) + \frac{1}{2}\ln\left(\frac{|W_m^\top \Sigma_\mathrm{signal}  W_m|}{| W_m^\top \Sigma_{Y|\theta} W_m |}\right)- \frac{1}{2}\ln\left( \frac{|\Sigma_\mathrm{signal}  |}{|\Sigma_{Y|\theta}|}\right). \label{eq:Bound_EIG_conservative_up}
 \end{align}
\end{corollary}
\begin{proof}
 Apply Theorem \ref{th:Bound_EIG} with $[W_m,W_\mathrm{new}]=I_p$, see details in Section \ref{proof:Bound_EIG_conservative}.
\end{proof}

The two corollaries above naturally yield two distinct strategies. When considering the incremental bounds \eqref{eq:Bound_EIG_incremental}--\eqref{eq:Bound_EIG_incremental_up}, a natural approach is to identify the design $W_m$ as the maximizer of the lower bound \eqref{eq:Bound_EIG_incremental}, namely
\begin{equation}\label{eq:incremental_optimization}
 \mathrm{(\texttt{i-SNR})}\qquad
 \max_{W_m\in\mathcal{K}_m} \frac{1}{2} \ln \left( \frac{|  W_m^\top \Sigma_\mathrm{signal} W_m    | }{ |  W_m^\top  \Sigma_{Y|\theta}  W_m |}\right).
\end{equation}
The upper bound \eqref{eq:Bound_EIG_incremental_up} can then be used to assess how sub-optimal the solution to \eqref{eq:incremental_optimization} is with respect to the original problem \eqref{eq:EIG}. Alternatively, the conservative bounds \eqref{eq:Bound_EIG_conservative}--\eqref{eq:Bound_EIG_conservative_up} naturally identify $W_m$ as the maximizer of the lower bound \eqref{eq:Bound_EIG_conservative}, that is
\begin{equation}\label{eq:conservative_optimization}
\mathrm{(\texttt{c-SNR})}\qquad
 \max_{W_m\in\mathcal{K}_m} \frac{1}{2} \ln \left( \frac{|  W_m^\top \Sigma_Y W_m    | }{ |  W_m^\top  \Sigma_\mathrm{noise}  W_m |}\right).
\end{equation}
Interestingly, the incremental strategy (\texttt{i-SNR}) is equivalent to maximizing the conservative upper bound \eqref{eq:Bound_EIG_conservative_up}, and vice versa.

\begin{remark}[Tightness of the bounds]\label{rmk:GaussianY}
 It is worth noting that both the incremental bounds \eqref{eq:Bound_EIG_incremental}--\eqref{eq:Bound_EIG_incremental_up} and the conservative bounds \eqref{eq:Bound_EIG_conservative}--\eqref{eq:Bound_EIG_conservative_up} become equalities as soon as
 \begin{equation}\label{eq:tightness}
   \Sigma_\mathrm{signal}=\Sigma_Y
   \quad\mathrm{and}\quad
   \Sigma_\mathrm{noise} = \Sigma_{Y|\theta} ,
 \end{equation}
 where we recall that $\Sigma_Y=\Cov(Y) $, $\Sigma_{Y|\theta} = \E[\Cov(Y|\theta)]$ and
 $\Sigma_\mathrm{signal}\preceq\mathcal{I}_Y^{-1}$,  $\Sigma_\mathrm{noise}\preceq \E[\mathcal{I}_{Y|\theta}]^{-1}$.
 As explained in Section \ref{sec:Fisher}, any random variable $X$ with smooth enough density satisfies the Cramér-Rao bound $\mathcal{I}_X^{-1} \preceq \Cov(X)$, with equality if and only if $X$ is Gaussian (\emph{Stein's characterization} of Gaussian distributions).
 We deduce that if \eqref{eq:tightness} holds true, then we have $\mathcal{I}_Y^{-1} = \Cov(Y)$ so that $Y$ is necessarily Gaussian. Similarly we have that $Y|\theta$ is almost surely Gaussian\footnote{By Jensen and Cramér-Rao inequalities, we have $\E[\mathcal{I}_{Y|\theta}]^{-1}\overset{\mathcal{J}}{\preceq} \E[\mathcal{I}_{Y|\theta}^{-1}] \overset{\mathcal{CR}}{\preceq} \E[\Cov(Y|\theta)]$, therefore \eqref{eq:tightness} implies $0=\E[\underbrace{\Cov(Y|\theta)-\mathcal{I}_{Y|\theta}^{-1}}_{\succeq0}]$, therefore $\mathcal{I}_{Y|\theta}^{-1}=\Cov(Y|\theta)$ a.s. and then $Y|\theta$ is a.s. Gaussian.}.
 The sharpness of the bounds therefore directly reflects the degree of non-Gaussianity of both $Y$ and $Y|\theta$.
\end{remark}

The following proposition shows that solving \eqref{eq:incremental_optimization} or \eqref{eq:conservative_optimization} in place of \eqref{eq:EIG} yields \emph{quasi-optimal solutions}, in the sense that the resulting EIG is comparable to the optimal EIG up to an additive term which directly relates to some divergence between $ \Sigma_\mathrm{signal}\leftrightarrow \Sigma_Y$ and between $\Sigma_\mathrm{noise} \leftrightarrow \Sigma_{Y|\theta}$.

\begin{proposition}[Quasi-optimality]\label{prop:quasi_opt_cons}
 Under the notations of Theorem \ref{th:Bound_EIG}, we denote by $\alpha_1\geq\hdots\geq\alpha_d $  and $\beta_1\geq\hdots\geq\beta_d $ the eigenvalues of the following generalized eigenvalue problems
 \begin{align*}
  \Sigma_Y u_i &=  \alpha_i \Sigma_\mathrm{signal} u_i \\
  \Sigma_{Y|\theta}  v_i &= \beta_i  \Sigma_\mathrm{noise}  v_i  ,
 \end{align*}
 where $u_i,v_i\in\R^p$ are the corresponding eigenvectors. Because $\Sigma_Y\succeq\Sigma_\mathrm{signal}$ and $\Sigma_{Y|\theta}\succeq\Sigma_\mathrm{noise}$ we have $\alpha_i\geq1$ and $\beta_i\geq1$.
 Then, for any $m\leq d$ we have
 \begin{align}
   \mathrm{EIG}( W_m^{\mathrm{inc}} ) &\geq \max_{W_m\in \mathcal{K}_m\subseteq\, \R^{p\times m}} \mathrm{EIG}(W_m) - \sum_{i=1}^m\frac{\ln(\alpha_i) + \ln(\beta_i)}{2}   \label{eq:quasi_opt_inc}\\
   \mathrm{EIG}( W_m^{\mathrm{cons}} ) &\geq \max_{W_m\in \mathcal{K}_m\subseteq\, \R^{p\times m}} \mathrm{EIG}(W_m) - \sum_{i=1}^{d-m}\frac{\ln(\alpha_i) + \ln(\beta_i)}{2}  ,  \label{eq:quasi_opt_cons}
 \end{align}
 where $W_m^{\mathrm{inc}}$ and $W_m^{\mathrm{cons}}$ are solutions to \eqref{eq:incremental_optimization} and \eqref{eq:conservative_optimization}, respectively.
\end{proposition}

\begin{proof}
 See Section \ref{proof:quasi_opt_cons}.
\end{proof}

Note that the suboptimality constant (\emph{i.e.}, the last term in \eqref{eq:quasi_opt_inc}) increases with $m$ for the incremental approach but decreases for the conservative approach.
This highlights the complementary regimes of the two approaches: the incremental approach is better suited to small values of $m$, while the conservative approach is preferable for large values of $m$.

\section{Computing the diagnostic matrices}\label{sec:ComputeDiagnostic}

We discuss different strategies for computing the diagnostic matrices $\Sigma_Y,\Sigma_{Y|\theta},\Sigma_\mathrm{signal}$ and $\Sigma_\mathrm{noise}$ in the goal-oriented setting, depending on the problem settings.
The following proposition gives closed form expressions for the diagnostic matrices in the nonlinear goal-oriented setting \eqref{eq:GOsetting}.

\begin{proposition}\label{prop:ComputeFisher}
 Consider the goal-oriented setting $Y=G(\eta) + \varepsilon$ and $\theta = H(\eta)+\xi$,
 where $G:\R^q\rightarrow\R^p$ and $H:\R^q\rightarrow\R^d$ are any measurable functions and where $\eta\sim\pi_{\eta}$ and $\varepsilon\sim\mathcal{N}(0,\Sigma_\mathrm{obs})$ and $\xi\sim\mathcal{N}(0,\Sigma_\xi)$ are assumed independent. Then
 \begin{align}\label{eq:ComputeFisher}
\begin{array}{rl}
   \Sigma_Y &= \Sigma_\mathrm{obs} + \Cov( G(\eta) )  \\
   \Sigma_{Y|\theta} &= \Sigma_\mathrm{obs} + \E[ \Cov( G(\eta) |\theta) ] \\
   \E[\mathcal{I}_{Y|\theta}] &= \Sigma_\mathrm{obs}^{-1}  - \Sigma_\mathrm{obs}^{-1}\E[ \Cov( G(\eta) |\theta,Y)]  \Sigma_\mathrm{obs}^{-1}\\
   \mathcal{I}_{Y} &= \Sigma_\mathrm{obs}^{-1} - \Sigma_\mathrm{obs}^{-1} \E[ \Cov( G(\eta) |Y)]\Sigma_\mathrm{obs}^{-1} .
\end{array}
\end{align}

\end{proposition}
\begin{proof}
 See Section \ref{proof:ComputeFisher}.
\end{proof}

In all above expressions, the diagnostic matrices involve expectations of conditional covariances of $G(\eta)$ in the form of $\E[\Cov(G(\eta)|Z)]$ where $Z \in \big\{\emptyset, \theta,(\theta,Y),Y \big\}$.
In the next subsections, we show how numerically estimate these quantities.

\begin{remark}[Linear \& Gaussian goal-oriented setting]\label{rmk:GOsettinglinear}
If we assume that $G(\eta)=A\eta$ and $H(\eta)=B\eta$ are linear and $\eta\sim\mathcal{N}(m_\eta,\Sigma_\eta)$, $\varepsilon\sim\mathcal{N}(0,\Sigma_\mathrm{obs})$ and $\xi\sim\mathcal{N}(0,\Sigma_\xi)$ are independent, we have that $(Y,\theta,\eta,\varepsilon)$ remains jointly Gaussian. Thus \eqref{eq:ComputeFisher} simplifies as
 \begin{align}\label{eq:GOsettinglinear_diag}
\begin{array}{rl}
   \Sigma_Y  &= \Sigma_\mathrm{obs} + A\Sigma_\eta A^\top   \\
  \Sigma_{Y|\theta} &= \Sigma_\mathrm{obs} + A (\Sigma_\eta^{-1} + B^\top \Sigma_\xi^{-1} B )^{-1} A^\top \\
  \E[\mathcal{I}_{Y|\theta}]^{-1} &= \Sigma_\mathrm{obs} + A (\Sigma_\eta^{-1} + B^\top \Sigma_\xi^{-1} B )^{-1} A^\top \\
  \mathcal{I}_{Y}^{-1} &=  \Sigma_\mathrm{obs} + A\Sigma_\eta A^\top  .
\end{array}
\end{align}
\end{remark}

\subsection{Sample-based diagnostic matrices}

We start by estimating the diagnostic matrix $\Sigma_Y=\Sigma_\mathrm{obs}+\Cov(G(\eta))$.
For simplicity, we use the notation $(\cdot)^{\otimes 2}$ for the outer product so that $ v^{\otimes 2} := v v^\top \in\R^{p\times p}$ for any vector $v\in\R^p$.
While a standard sample-based estimator of the covariance matrix $\Cov(G(\eta))$ can be used, we propose here to exploit the formula $\Cov(G(\eta))=\frac{1}{2}\E[(G(\eta)-G(\eta'))^{\otimes 2}]$, where $\eta'$ is an independent copy of $\eta$. This leads to the following \emph{unbiased} estimator
\begin{equation}\label{eq:estimator_SigmaY}
 \Sigma_Y^{(N)} = \Sigma_\mathrm{obs} + \frac{1}{2N} \sum_{i=1}^N  \Big( G(\eta^{(i)}) - G(\eta'^{(i)}) \Big)^{\otimes 2}  ,
 \qquad \mathrm{where }
 \begin{cases}
  \eta^{(i)}  \sim \pi_\eta \\
  \eta'^{(i)}  \sim \pi_\eta.
 \end{cases}
\end{equation}

Estimating $\Sigma_{Y|\theta}=\Sigma_\mathrm{obs}+\E[\Cov(G(\eta)|\theta)]$ can be done the same way by redefining $\eta'$ such that $(\eta',\eta,\theta)$ has joint density $\pi_{\eta',\eta,\theta}(\eta',\eta,\theta) := \pi_{\eta|\theta}(\eta'|\theta)  \pi_{\eta,\theta}(\eta,\theta) $. With this definition, $\eta'$ and $\eta$ are \emph{conditionally independent} given $\theta$ so that we can write  $\E[ \Cov(G(\eta)|\theta)] = \E[ \frac{1}{2}\E[( G(\eta) - G(\eta') )^{\otimes 2}|\theta]] = \frac{1}{2} \E[( G(\eta) - G(\eta') )^{\otimes 2}]$.
This leads to the following \emph{unbiased} estimator
\begin{equation}\label{eq:estimator_SigmaYTheta}
 \Sigma_{Y|\theta}^{(N)} = \Sigma_\mathrm{obs} + \frac{1}{2N} \sum_{i=1}^N  \Big( G(\eta^{(i)}) - G(\eta'^{(i)}) \Big)^{\otimes 2}  ,
 \quad \mathrm{where }
 \begin{cases}
  \eta^{(i)}  \sim \pi_\eta \\
  \theta^{(i)} \sim \pi_{\theta|\eta}(\cdot|\eta^{(i)}) \\
  \eta'^{(i)}  \sim \pi_{\eta|\theta}(\cdot|\theta^{(i)}) .
 \end{cases}
\end{equation}
While $\eta^{(i)}  \sim \pi_\eta$ and  $\theta^{(i)} \sim \pi_{\theta|\eta}(\cdot|\eta^{(i)})$ are readily generated via \eqref{eq:GOsetting}, drawing $\eta'^{(i)}  \sim \pi_{\eta|\theta}(\cdot|\theta^{(i)})$ can be more challenging. This is the bottleneck of the estimator \eqref{eq:estimator_SigmaYTheta}.

The problem of sampling from $\pi_{\eta|\theta}(\cdot|\theta^{(i)})$ interprets as a Bayesian inverse problem on $\theta = H(\eta)+\xi$. In principle, one can apply any MCMC algorithm targeting $\pi_{\eta|\theta}(\cdot|\theta^{(i)})$ where each iteration requires the evaluation of the prior $\pi_\eta(\cdot)$ and of the likelihood function $\pi_{\theta|\eta}(\theta^{(i)}|\cdot)\propto\exp(-\frac{1}{2}\|H(\cdot)-\theta^{(i)}\|_{\Sigma_\xi^{-1}}^2)$. For the applications we have in mind, evaluating the function $H$ is inexpensive (compared to evaluating the forward model $G$), therefore we neglect the cost of sampling from $\pi_{\eta|\theta}(\cdot|\theta^{(i)})$.
In some cases, samples from $\pi_{\eta|\theta}(\cdot|\theta^{(i)})$ can be obtained in closed form, as shown by the following remarks.

\begin{remark}[Pick-freeze estimators]
 Another relevant scenario where $\pi_{\eta|\theta}(\cdot|\theta^{(i)})$ is available in closed form is when $\theta = \eta_1$ is defined as one component (or a block of components) of $\eta = (\eta_1,\eta_2)$, and when $\eta_1$ and $\eta_2$ are independent. In that case $\eta|\theta = \eta|\eta_1 $ reduces to $ (\eta_1,\eta_2')$, where $\eta_2'\sim\pi_{\eta_2}$, and yields the following estimator
\begin{equation}\label{eq:estimator_SigmaYTheta_PF}
 \Sigma_{Y|\theta}^{(N)} = \Sigma_\mathrm{obs} + \frac{1}{2N} \sum_{i=1}^N  \Big( G(\eta_1^{(i)},\eta_2^{(i)}) - G(\eta_1^{(i)},\eta_2'^{(i)}) \Big)^{\otimes 2}  ,
 \qquad \mathrm{where }
 \begin{cases}
  \eta_1^{(i)}  \sim \pi_{\eta_1} \\
  \eta_2^{(i)}  \sim \pi_{\eta_2} \\
  \eta_2'^{(i)}  \sim \pi_{\eta_2}.
 \end{cases}
\end{equation}
 This estimator is called the \emph{pick-freeze} estimator and is widely used in Global sensitivity Analysis, see \cite{janon2014asymptotic,gamboa2016statistical}.
\end{remark}

While, in principle, sample-based estimators can also be used to compute the conditional expectations $\E[\Cov(G(\eta)| Y)]$ and $\E[\Cov(G(\eta)| Y,\theta)]$, as in \eqref{eq:estimator_SigmaYTheta}, this approach is not viable in practice. The difficulty lies in the need to generate the samples $\eta'^{(i)}\sim \pi_{\eta| Y}(\cdot| Y^{(i)})$ (or $\pi_{\eta | \theta,Y}(\cdot | \theta^{(i)},Y^{(i)})$) for $i=1,\hdots,N$. This would require running MCMC algorithms $N$ times, whose computational cost is prohibitive since each MCMC iteration will involves evaluating the forward model $G$. For this reason, we propose in the next subsections two alternative strategies.

\subsection{Approximation-based diagnostic matrices}
\label{sec:Approxiamtion_based_diagnostics}

In this section we estimate $\E[\Cov(G(\eta)| Y)]$ and $\E[\Cov(G(\eta)| Y,\theta)]$ by approximating the conditional expectations $\E[G(\eta)|Y]$ and $\E[G(\eta)|Y,\theta]$.
Recall that $\E[G(\eta)|Y]$ is defined as the orthogonal projection of $G(\eta)$ onto the Hilbert space $L^2_{\pi_Y}(\R^p):=\{f:\R^p\rightarrow\R^p: \E[\|f(Y)\|^2]<\infty \}$ of square integrable functions.
Thus, given an approximation space $\mathcal{F}\subset L^2_{\pi_Y}(\R^p)$, the function
\begin{equation}\label{eq:conditionalExpectation_Y}
 f\in \argmin_{f\in \mathcal{F}} \E[\| G(\eta) - f(Y)\|_2^2] ,
\end{equation}
is an approximation of the conditional expectation $y\mapsto\E[G(\eta)|Y=y]$.
By construction, $f$ approximately removes the noise $\varepsilon$, in the sense that $f( G(\eta)+\varepsilon ) \approx G(\eta)$, hence $f$ is commonly referred to as a \emph{denoiser} in machine learning.

Once the function $f$ solution to \eqref{eq:conditionalExpectation_Y} (or to a sampled-based discretization of it) is computed, one can approximate the expected conditional covariance as follow
\begin{align*}
 \E[\Cov(G(\eta)| Y)]
 &= \E[( G(\eta) - \E[G(\eta)|Y] )^{\otimes 2}] \\
 &= \E[( G(\eta) - f(Y))^{\otimes 2}] -\E[( f(Y) - \E[G(\eta)|Y] )^{\otimes 2}] \\
 &\preceq \E[( G(\eta) - f(Y) )^{\otimes 2}] \\
 &\approx \frac{1}{N} \sum_{i=1}^N  ( G(\eta^{(i)}) - f(Y^{(i)}) )^{\otimes 2} ,
\end{align*}
where $(\eta^{(i)},Y^{(i)})$ are independent samples from $\pi_{\eta,Y}$ generated via \eqref{eq:GOsetting}.
It is worth to notice that a similar approach has been recently employed in \cite{hoang2025scalable}, albeit for approximating $\E[\theta|Y_m]$ in order to perform A-optimal experimental design.

In the above estimator, there are three sources of error: (i) the sample error, which decreases as $N\rightarrow\infty$, (ii) the reconstruction error of the conditional expectation by an element in $\mathcal{F} \subsetneq L^2_{\pi_Y}(\R^p)$ and (iii) the discretization error introduced by the numerical solution of \eqref{eq:conditionalExpectation_Y} via \emph{e.g.} discrete least-squares.
Achieving an appropriate balance between these error contributions is critical for the overall efficiency of the method. Such detailed error analysis, however, is beyond the scope of the present work and is therefore left for future investigation.

Estimating $\E[\Cov(G(\eta)| Y,\theta)]$ can be done in the same way by introducing
\begin{equation}\label{eq:conditionalExpectation_Ytheta}
 g\in \argmin_{g\in \mathcal{G}} \E[\| G(\eta) - g(Y,\theta)\|_2^2] ,
\end{equation}
for some $\mathcal{G}\subset L^2_{\pi_{Y,\theta}}(\R^p):=\{g:\R^p\rightarrow\R^p: \E[\|g(Y,\theta)\|^2]<\infty \}$.
Together with the function $f$, this naturally yields the following diagnostic matrices
\begin{align}
 \Sigma_\mathrm{signal}^{(N,\mathcal{F})} &= \left(\Sigma_\mathrm{obs}^{-1} - \Sigma_\mathrm{obs}^{-1} \left(  \frac{1}{N} \sum_{i=1}^N  ( G(\eta^{(i)}) - f(Y^{(i)}) )^{\otimes 2}  \right)\Sigma_\mathrm{obs}^{-1}\right)^{-1} \label{eq:SigmaSignal_Approx} \\
 \Sigma_\mathrm{noise}^{(N,\mathcal{G})} &=  \left(\Sigma_\mathrm{obs}^{-1} - \Sigma_\mathrm{obs}^{-1} \left(  \frac{1}{N} \sum_{i=1}^N  ( G(\eta^{(i)}) - g(Y^{(i)},\theta^{(i)}) )^{\otimes 2}  \right)\Sigma_\mathrm{obs}^{-1}\right)^{-1} .
 \label{eq:SigmaNoise_Approx}
\end{align}
In order to be well defined, the expressions inside the parenthesis above must be invertible.
When considering $\Sigma_\mathrm{signal}^{(N,\mathcal{F})}$ we need to ensure $\frac{1}{N} \sum_{i=1}^N  ( G(\eta^{(i)}) - f(Y^{(i)}) )^{\otimes 2} \prec \Sigma_\mathrm{obs}$ for large enough $N$.
A necessary condition for this is $\E[(G(\eta) - f(Y) )^{\otimes 2}] \prec \Sigma_\mathrm{obs}$.
The next proposition shows that this condition is satisfied whenever $\mathcal{F}$ and $\mathcal{G}$ contain all affine functions.
Finally, standard concentration inequalities imply that $\Sigma_\mathrm{signal}^{(N,\mathcal{F})}$ (and similarly $\Sigma_\mathrm{noise}^{(N,\mathcal{G})}$) is well defined with high probability, provided $N$ is sufficiently large. We leave the detailed analysis of these finite-sample conditions to future work.

\begin{proposition}\label{prop:Sigma_FG}
 Let $\mathcal{F} \subset L^2_{\pi_{Y}}(\R^p)$ and $\mathcal{G}\subset L^2_{\pi_{Y,\theta}}(\R^p)$ be such that $f\in\mathcal{F} \Rightarrow (\alpha f+\ell)\in\mathcal{F}$ and  $g\in\mathcal{G} \Rightarrow (\alpha g+\ell)\in\mathcal{G}$ for any scalar $\alpha\in\R$ and any affine function $\ell$. Let $f$ and $g$ be defined by \eqref{eq:conditionalExpectation_Y} and \eqref{eq:conditionalExpectation_Ytheta}, respectively.
 If $\Cov(G(\eta))\succ0$ and $\Sigma_\mathrm{obs}\succ0$, then $\E[(G(\eta) - f(Y) )^{\otimes 2}] \prec \Sigma_\mathrm{obs}$ and $\E[(G(\eta) - g(Y,\theta) )^{\otimes 2}] \prec \Sigma_\mathrm{obs}$ so that
 \begin{align*}
  \Sigma_\mathrm{signal}^{(\infty,\mathcal{F})} &= \left(\Sigma_\mathrm{obs}^{-1} - \Sigma_\mathrm{obs}^{-1}  \E[(G(\eta) - f(Y) )^{\otimes 2}] \Sigma_\mathrm{obs}^{-1}\right)^{-1}  \\
  \Sigma_\mathrm{noise}^{(\infty,\mathcal{G})} &=\left(\Sigma_\mathrm{obs}^{-1} - \Sigma_\mathrm{obs}^{-1}  \E[(G(\eta) - g(Y,\theta) )^{\otimes 2}] \Sigma_\mathrm{obs}^{-1}\right)^{-1}  ,
 \end{align*}
 are well defined.
\end{proposition}

\begin{proof}
 See Section \ref{proof:Sigma_FG}.
\end{proof}

The strength of these estimators is that they are \emph{consistent} in the sense that $\Sigma_\mathrm{signal}^{(N,\mathcal{F})}\rightarrow (\mathcal{I}_Y)^{-1}$ and $\Sigma_\mathrm{noise}^{(N,\mathcal{G})}\rightarrow (\E[\mathcal{I}_{Y|\theta}])^{-1}$ almost surely when $N\rightarrow\infty$ and when $f(Y)\rightarrow\E[G(\eta)|Y]$, $g(Y,\theta)\rightarrow\E[G(\eta)|Y,\theta]$ in $L^2$. However, they are such that
\begin{align*}
 \left(\Sigma_\mathrm{signal}^{(N,\mathcal{F})}\right)^{-1} \approx \left(\Sigma_\mathrm{signal}^{(\infty,\mathcal{F})}\right)^{-1} \preceq \mathcal{I}_Y
 \quad\mathrm{and}\quad
 \left(\Sigma_\mathrm{noise}^{(N,\mathcal{G})}\right)^{-1} \approx \left(\Sigma_\mathrm{noise}^{(\infty,\mathcal{G})}\right)^{-1} \preceq \E[\mathcal{I}_{Y|\theta}] ,
\end{align*}
where the inequalities above are in the opposite direction compared to \eqref{eq:SigmaNoise} and \eqref{eq:SigmaSignal}.
As a result, they cannot be safely used in Theorem \ref{th:Bound_EIG}, unless $\mathcal{F}$ and $\mathcal{G}$ are sufficiently rich to accurately recover the conditional expectations, which can be difficult to guarantee in practice. In our numerical experiments, we did not encountered such difficulties. In the next section, we introduce alternative diagnostic matrices which provide bounds on the Fisher information in the required direction.

\subsection{Gradient-based diagnostic matrices}

We have the following result.

\begin{proposition}[Gradient-based diagnostic matrices]\label{prop:BoundFisher}
 Consider the setting $Y=G(\eta)+\varepsilon$ and $\theta = H(\eta)+\xi$,
 where $G:\R^q\rightarrow\R^p$ and $H:\R^q\rightarrow\R^d$ are continuously differentiable functions and where $\eta\sim\pi_{\eta}$ and $\varepsilon\sim\mathcal{N}(0,\Sigma_\mathrm{obs})$ and $\xi\sim\mathcal{N}(0,\Sigma_\xi)$ are independent.
 Let
  \begin{align}
 \mathcal{L}(G) &=  \E[ \Jac  G(\eta) ]  \label{eq:Lu}\\
 \mathcal{H}(G)
 &=  \E\left[ \Big( \Jac  G(\eta) -\mathcal{L}(G) \Big)^\top \Sigma_\mathrm{obs}^{-1} \Big( \Jac  G(\eta) -\mathcal{L}(G)\Big) \right]  \label{eq:Hu} \\
 \mathcal{J}(H) &= \E[  \Jac  H(\eta)^\top \Sigma_\xi^{-1}\Jac  H(\eta) ] \label{eq:Jh} \\
 \mathcal{I}_\eta &= \Cov( \nabla_\eta \ln\pi_\eta(\eta)) ,
\end{align}
where $\Jac G(\eta)\in\R^{p\times q}$ and $\Jac H(\eta)\in\R^{p\times d}$ are the Jacobian matrices of $G$ and $H$ evaluated at $\eta$.
Then the diagnostic matrices
\begin{align}
 \Sigma_\mathrm{signal} &= \Sigma_\mathrm{obs}   + \mathcal{L}(G)   \Big(  \mathcal{H}(G) + \mathcal{I}_\eta  \Big)^{-1} \mathcal{L}(G)^\top \label{eq:SigmaSignal_boundFisher} \\
 \Sigma_\mathrm{noise} &= \Sigma_\mathrm{obs}   + \mathcal{L}(G)   \Big(  \mathcal{H}(G) + \mathcal{I}_\eta + \mathcal{J}(H)  \Big)^{-1} \mathcal{L}(G)^\top , \label{eq:SigmaNoise_boundFisher}
\end{align}
 satisfy $\Sigma_\mathrm{signal}^{-1}\succeq \mathcal{I}_{Y}$ and $\Sigma_\mathrm{noise}^{-1}\succeq \E[\mathcal{I}_{Y|\theta}]$.

\end{proposition}
\begin{proof}
 The proof consists essentially in applying Proposition \ref{prop:Fisher}, see details in Section \ref{proof:BoundFisher}.
\end{proof}

Proposition \ref{prop:BoundFisher} introduces two diagnostic matrices $\Sigma_\mathrm{signal}$ and $\Sigma_\mathrm{noise}$ which satisfy the conditions \eqref{eq:SigmaNoise} and \eqref{eq:SigmaSignal} of Theorem \ref{th:Bound_EIG}. These matrices can be readily approximated using Monte Carlo sampling and evaluations of the Jacobians of $H$ and $G$, see Section \ref{sec:certification}.
Note that \eqref{eq:SigmaSignal_boundFisher} and \eqref{eq:SigmaNoise_boundFisher} resemble the expressions obtained in the linear setting $G(\eta)=A\eta$ and $H(\eta) = B\eta$ of Remark \ref{rmk:GOsettinglinear} for which
$\E[\mathcal{I}_{Y|\theta}]^{-1} = \Sigma_\mathrm{obs} + A (\Sigma_\eta^{-1} + B^\top \Sigma_\xi^{-1} B )^{-1} A^\top $ and  $\mathcal{I}_{Y}^{-1} =  \Sigma_\mathrm{obs} + A\Sigma_\eta A^\top $, remember \eqref{eq:GOsettinglinear_diag}.

\begin{remark}[Best affine approximation to the forward model]
 Consider the problem of finding the matrix $A\in\R^{q\times p}$ and the vector $b\in\R^p$ which minimize the mean square error $\E\left[\| G(\eta)- (A\eta+b) \|^2\right]$. It is well known that the solution is $b^\star=\E[G(\eta)]-Am$ and $A^\star=\E[G(\eta) (\eta-m)^\top]\Sigma^{-1}$, where $m$ and $\Sigma$ are the mean and the covariance of $\eta$. If we further assume that $\eta=\mathcal{N}(m,\Sigma)$ is Gaussian, one integration by part yields
 $$
  A^\star
  =\int G(\eta) (\underbrace{\Sigma^{-1}(\eta-m)}_{-\nabla\ln\pi_\eta(\eta)})^\top \pi_\eta(\eta) \d\eta = \int \Jac G(\eta) \pi_\eta(\eta) \d\eta = \mathcal{L}(G).
 $$
 This means that $\mathcal{L}(G)$ is the slope of the best affine approximation of $G$ in $L^2_{\pi_\eta}$.
\end{remark}

It is worth emphasizing that the matrix $\mathcal{H}(G)$ quantifies the mean-square deviation of $\Jac G(\eta)$ from its mean $\mathcal{L}(G)$ (somehow the covariance of $\Jac G(\eta)^\top$). In particular, $\mathcal{H}(G)$ provides a measure of the non-affinity of $G$: for continuously differentiable $G$, $\mathcal{H}(G)=0$ implies that $G$ is affine. Note that the larger $\mathcal{H}(G)$ is, the closer $\Sigma_\mathrm{signal}$ and $\Sigma_\mathrm{noise}$ are to $\Sigma_\mathrm{obs}$.

\begin{corollary}
 Assume the data generating process $Y=G(\theta,\eta)+\varepsilon$ where $G:\R^d\times\R^q\rightarrow\R^p$ is continuously differentiable and $\theta\sim\pi_\theta$, $\eta\sim\pi_\eta$ and $\varepsilon\sim\mathcal{N}(0,\Sigma_\mathrm{obs})$ are independent such that $\mathcal{I}_\eta = \Cov( \nabla_\eta \ln\pi_\eta(\eta))$ and $\mathcal{I}_{\theta,\eta} = \Cov( \nabla_{\theta,\eta} \ln\pi_{\theta,\eta}(\theta,\eta))$ are well defined. Let
 $\mathcal{L}_\eta(G) =  \E[ \Jac_\eta  G(\theta,\eta) ]$, $\mathcal{L}_{\theta,\eta}(G) =  \E[ \Jac_{\theta,\eta}  G(\theta,\eta) ]$ and
 \begin{align*}
 \mathcal{H}_\eta(G) &=  \E\left[ \Big( \Jac_\eta  G(\theta,\eta) -\mathcal{L}_\eta(G) \Big)^\top \Sigma_\mathrm{obs}^{-1} \Big( \Jac_\eta  G(\theta,\eta) -\mathcal{L}_\eta(G)\Big) \right]  \\
 \mathcal{H}_{\theta,\eta}(G) &=  \E\left[ \Big( \Jac_{\theta,\eta}  G(\theta,\eta) -\mathcal{L}_{\theta,\eta}(G) \Big)^\top \Sigma_\mathrm{obs}^{-1} \Big( \Jac_{\theta,\eta}  G(\theta,\eta) -\mathcal{L}_{\theta,\eta}(G)\Big) \right] ,
\end{align*}
Then
\begin{align}
 \Sigma_\mathrm{signal} &= \Sigma_\mathrm{obs}   + \mathcal{L}_{\theta,\eta}(G)   \Big(  \mathcal{H}_{\theta,\eta}(G) + \mathcal{I}_{\theta,\eta} \Big)^{-1} \mathcal{L}_{\theta,\eta}(G)^\top \label{eq:SigmaSignal_boundFisher_Auxiliary} \\
 \Sigma_\mathrm{noise} &= \Sigma_\mathrm{obs}   + \mathcal{L}_{\eta}(G)   \Big(  \mathcal{H}_{\eta}(G) + \mathcal{I}_{\eta} \Big)^{-1} \mathcal{L}_{\eta}(G)^\top , \label{eq:SigmaNoise_boundFisher_Auxiliary}
\end{align}
satisfy $\Sigma_\mathrm{signal}^{-1}\succeq \mathcal{I}_{Y}$ and $\Sigma_\mathrm{noise}^{-1}\succeq \E[\mathcal{I}_{Y|\theta}]$.
\end{corollary}

\begin{proof}
 Use $H(\eta_1,\eta_2) = \eta_1$ and $\xi=\mathcal{N}(0,\sigma_\xi^2 I_d)$ Proposition \ref{prop:BoundFisher} and let $\sigma_\xi^2\rightarrow0$.
\end{proof}

\subsection{Summary}

Table \ref{tab:formula} summarizes the different ways to compute the four diagnostic matrices in practice.

\begin{table}[h]
\centering  \footnotesize
\begin{tabular}{|>{\centering\arraybackslash}p{3cm}|c|c|c|c|}
\cline{2-5}
\multicolumn{1}{c|}{} & \multicolumn{2}{c|}{Incremental SNR \eqref{eq:incremental_optimization}} & \multicolumn{2}{c|}{Conservative SNR \eqref{eq:conservative_optimization}} \\[0.1cm]
\multicolumn{1}{c|}{}  &  $\Sigma_\mathrm{signal} \preceq \mathcal{I}_Y^{-1} $ &  $\Sigma_{Y|\theta} = \E[\Cov(Y|\theta)]$ & $\Sigma_Y=\Cov(Y)$ & $\Sigma_\mathrm{noise}\preceq\E[\mathcal{I}_{Y|\theta}]$ \\[0.1cm] \hline
\textbf{Standard}  &Approx: \eqref{eq:SigmaSignal_Approx} & $\Sigma_\mathrm{obs}$ (given)& Sample: \eqref{eq:estimator_SigmaY}&$\Sigma_\mathrm{obs}$ (given) \\ $Y=G(\theta)+\varepsilon $  &  Grad: \eqref{eq:SigmaSignal_boundFisher} &  &   &  \\ \hline
\textbf{Auxiliary}& Approx: \eqref{eq:SigmaSignal_Approx}& Sample: \eqref{eq:estimator_SigmaYTheta}& Sample: \eqref{eq:estimator_SigmaY}&Approx: \eqref{eq:SigmaNoise_Approx}\\ $Y=G(\theta,\eta)+\varepsilon $ &Grad: \eqref{eq:SigmaSignal_boundFisher_Auxiliary}&&&Grad: \eqref{eq:SigmaNoise_boundFisher_Auxiliary}  \\ \hline
\textbf{Goal-oriented} &Approx: \eqref{eq:SigmaSignal_Approx}&&&Approx: \eqref{eq:SigmaNoise_Approx} \\ $Y=G(\eta)+\varepsilon $,&Grad: \eqref{eq:SigmaSignal_boundFisher}&Sample:  \eqref{eq:estimator_SigmaYTheta}&Sample: \eqref{eq:estimator_SigmaY}&Grad: \eqref{eq:SigmaNoise_boundFisher}  \\ $\theta = H(\eta) + \xi$ &&&&   \\ \hline
\end{tabular}
\caption{Formula to compute $\Sigma_{Y}$, $\Sigma_{Y|\theta}$, $\Sigma_\mathrm{signal}$ and $\Sigma_\mathrm{noise}$ in practice.}
\label{tab:formula}
\end{table}

The primary computational advantage of the proposed SNR bounds is that they completely decouple the evaluation of the forward model $G$ from the design optimisation procedure.
Once the four $p \times p$ diagnostic matrices are assembled, finding the optimal design $W_m$ (e.g., via a greedy algorithm) reduces to purely algebraic operations, specifically, computing determinants of $m \times m$ matrices, requiring no further evaluations of $G$.
This contrasts with standard EIG estimators, which rely on nested Monte Carlo sampling and typically demand $\mathcal{O}(N^2)$ evaluations of $G$ for each candidate design.

Consequently, the cost of our method is entirely concentrated in the offline assembly of the diagnostic matrices (summarized in Table \ref{tab:formula}), which scale as $\mathcal{O}(N)$. The number of forward model evaluations for $N$ samples breaks down as follows:
\begin{itemize}
    \item Approximation-based diagnostics: Generating the training dataset requires $N$ evaluations of $G$. Training the surrogate functions $f$ and $g$ incurs an offline computational overhead but requires no additional forward solves.
    \item Sample-based diagnostics: The pick-and-freeze estimators (e.g., \eqref{eq:estimator_SigmaY} and \eqref{eq:estimator_SigmaYTheta_PF}) require strictly $2N$ evaluations of $G$ per matrix, corresponding to the paired samples $\eta$ and $\eta'$.
    \item Gradient-based diagnostics: Evaluating the bounds requires $N$ evaluations of the Jacobian $\Jac G(\eta)$. While this provides mathematical certification, its relative cost may become prohibitive (e.g., typically costing $\mathcal{O}(p)$ times the cost of $G$ using adjoint methods, or $\mathcal{O}(q)$ using forward sensitivities).
\end{itemize}

\section{Numerical illustration}
\label{sec:numerics}

The following numerical experiments are computed with code provided at \url{https://github.com/mohamed495/CBOED}.

\subsection{Problem setup}
\label{sec:test-case}

We consider the solution $u=u(t,x;\eta)$ to the viscous Burgers equation
\begin{equation}
	\partial_t u + \lambda u \partial_x u = \nu \partial_{xx} u,
	\qquad x \in \Omega = [0,1],\ t \in [0,T],
	\label{eq:burgers}
\end{equation}
with initial condition $u(0,x;\eta) = \eta(x)$ and Dirichlet boundary conditions $u(t,0;\eta)=u(t,1;\eta)=1$.
Here, the initial condition $\eta$ is modelled as a Gaussian process with constant mean $1$ and a Matérn-$3/2$ covariance kernel (with length scale $\ell=0.2$ and pointwise standard deviation $\sigma=0.3$), conditioned on the boundary values $\eta(0)=\eta(1)=1$.
The diffusivity parameter is set to $\nu = 0.003$ and the nonlinearity parameter $\lambda \ge 0$ interpolates between the linear heat equation ($\lambda = 0$) and the standard viscous Burgers equation ($\lambda = 1$).

We discretize \eqref{eq:burgers} in space using a finite-difference scheme on a uniform grid with $N=200$ interior points. Time integration is performed using a Crank--Nicolson scheme with time step $\Delta t=10^{-3}$ over $n_t=100$ time steps, yielding a final time $T=0.1$. Denoting by $u(t;\eta)\in\mathbb{R}^N$ the vector containing the spatial degrees of freedom at discrete time $t$, we define the forward map $G:\mathbb{R}^{q}\to\mathbb{R}^{p}$ as
\begin{equation}
G(\eta)=u(T;\eta),
\end{equation}
which maps the discretized initial condition $\eta$, now a $p$-dimensional Gaussian random vector $\mathcal{N}(m_\eta,\Sigma_\eta)$, to the numerical solution to \eqref{eq:burgers} at the time $t=T$. In this setting we have $q=p=N=200$.
The observable data are subject to additive Gaussian noise
\begin{equation}
	Y = G(\eta) + \varepsilon, \qquad
	\varepsilon \sim \mathcal{N}(0, \Sigma_{\rm obs}),
	\qquad \Sigma_{\rm obs} = 0.01 \times I_N .
\end{equation}
We consider the following definitions of the parameter of interest $\theta$.
\begin{itemize}
 \item \textbf{Standard setting} where $\theta=\eta$ corresponds to the entire initial condition. In this case, $d=N=200$.
 \item \textbf{Auxiliary parameter} where $\theta=(\eta_1,\dots,\eta_{100})$ corresponds to the first half of the initial condition, the other half being considered as a auxiliary parameter. Here, $d=100$.
 \item \textbf{Goal-oriented setting} where $\theta = \|\eta\|^2 $ is the $\ell^2$ norm squared of the initial condition. Here, $d=1$.
\end{itemize}

We compute the diagnostic matrices $\Sigma_{Y}$, $\Sigma_{Y|\theta}$, $\Sigma_\mathrm{signal}$ and $\Sigma_\mathrm{obs}$ using the formula given in Table \ref{tab:formula}. For each estimator, we use a sufficiently large number of samples (here, $N=10^4$ samples) so that the resulting sampling error can be neglected.
We consider point observation of the final state, which corresponds to the set of design $\mathcal{K}_m = \{ W_m=[e_{\tau_1},\hdots,e_{\tau_m}] : \tau\subset\{1,\hdots,p\}, \#\tau=m \}$.
Finally, we employ a greedy procedure to maximize the SNR ratio, that is
\begin{equation}\label{eq:Greedy}
W_{m+1} = [W_{m},w_{w+1}]
\quad\mathrm{where}\quad
w_{m+1} \in \argmax_{w\in\{e_1,\hdots,e_d\}} \mathrm{SNR}([W_m,w]) ,
\end{equation}
where $\mathrm{SNR}(W)=\frac{1}{2}\ln\frac{|W^\top A W|}{|W^\top B W|}$ denotes either the incremental SNR \eqref{eq:incremental_optimization} or the conservative SNR \eqref{eq:conservative_optimization}.
More sophisticated algorithms, not considered here, can also be employed, see \emph{e.g.} \cite{eswar2026bayesian}.

\subsection{Linear forward model}

We first consider $\lambda=0$ which corresponds to a linear forward model $G$.
In this linear-Gaussian setting, $\eta|Y_m$ remains Gaussian for any design $W_m$ and is known in closed form. In addition, for both the standard and auxiliary parameter settings, the EIG can be computed in closed form using the diagnostic matrices \eqref{eq:GOsettinglinear_diag}. Figure \ref{fig:LinearCase} represents prior and posterior realizations of the initial condition $\eta$. We observe that, in the goal-oriented setting $\theta =(\eta_1,\hdots, \eta_{100})$, the resulting sensor locations yield a well-identified initial condition within the region of interest (the left part, corresponding to the first 100 components of $\eta$), while little information is propagated to the right part of the domain.

\begin{figure}
    \centering
    \begin{subfigure}{0.48\textwidth}
        \includegraphics[width=\linewidth]{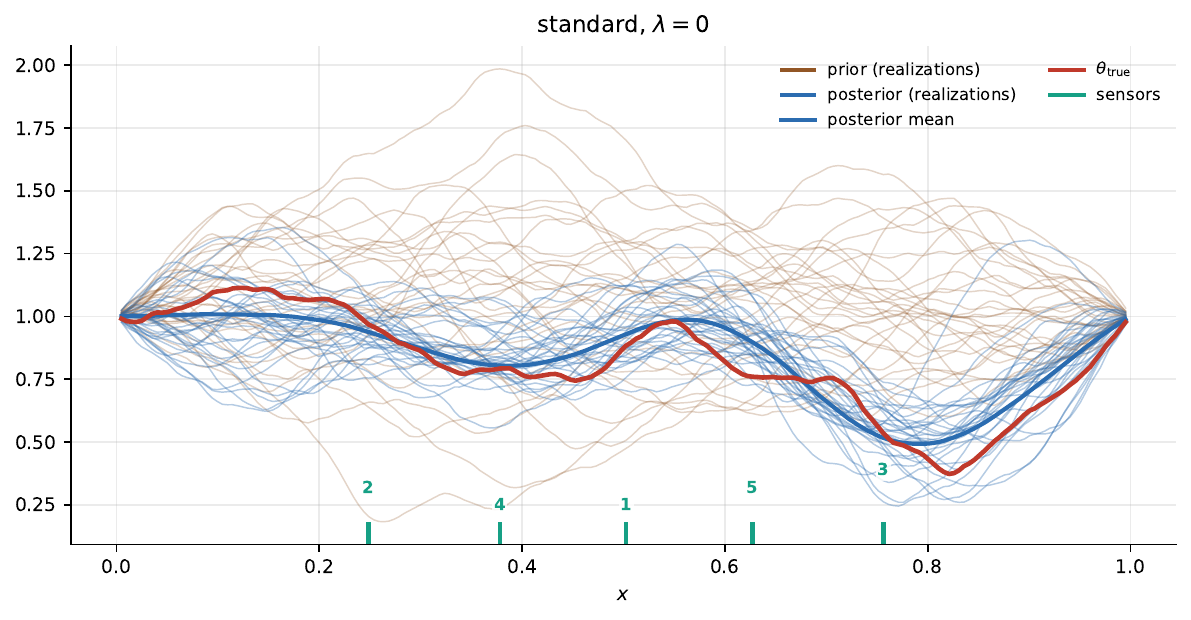}
        \caption{Standard setting}
        \label{fig:LinearCase_a}
    \end{subfigure}
    \begin{subfigure}{0.48\textwidth}
        \includegraphics[width=\linewidth]{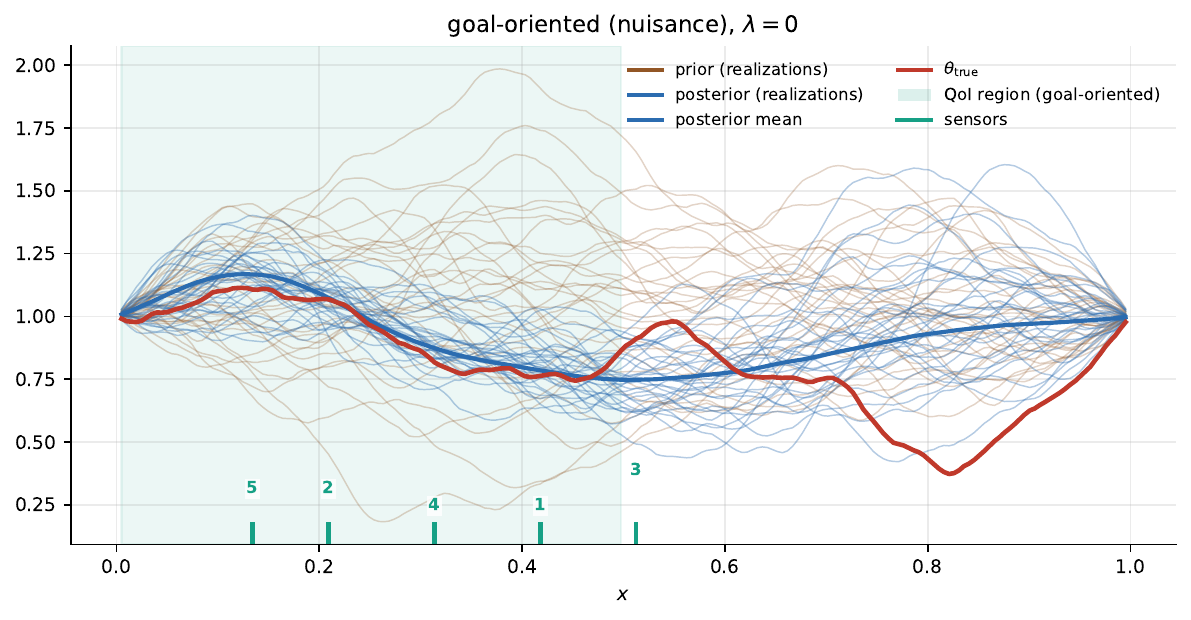}
        \caption{Auxiliary parameter}
        \label{fig:LinearCase_b}
    \end{subfigure}

    \caption{Linear forward model $\lambda=0$. Prior and posterior samples of the initial condition $\eta$ using a design of $m=5$ experiments obtained by a greedy algorithm \eqref{eq:Greedy} to optimize the EIG (closed form expression using \eqref{eq:GOsettinglinear_diag}). The location of the sensors is represented by the green dashes.}
    \label{fig:LinearCase}
\end{figure}

The goal-oriented setting $\theta=\|\eta\|^2$ presents a nonlinearity which does not permit to compute the EIG in closed form (because $\theta|Y_m$ is no longer Gaussian).
We consider here the gradient-based diagnostics matrices \eqref{eq:SigmaSignal_boundFisher} and \eqref{eq:SigmaNoise_boundFisher} given by Proposition \ref{prop:BoundFisher}. Because the forward model $G$ is linear, and by taking $H(\eta)=\|\eta\|^2$ and $\xi=0$, these matrices simplify to $\Sigma_\mathrm{signal} = \Cov(Y)$ and $\Sigma_\mathrm{noise} = \Sigma_\mathrm{obs}$. As a consequence, the resulting \texttt{c-SNR} problem \eqref{eq:conservative_optimization} is the same as the one obtained by the standard setting, namely $\max_{W_m}\frac{1}{2}\ln(\frac{|W_m^\top \Cov(Y) W_m|}{|W_m^\top \Sigma_\mathrm{obs} W_m|})$. The \texttt{i-SNR}, in turn, becomes $\max_{W_m}\frac{1}{2}\ln(\frac{|W_m^\top \Cov(Y) W_m|}{|W_m^\top \Sigma_{Y|\theta} W_m|})$, where $\Sigma_{Y|\theta}=\E[\Cov(Y|\theta)]$ can be estimated using \eqref{eq:estimator_SigmaYTheta}. This estimator requires samples from $\eta|\theta$, which we generate using a rejection algorithm: given a tolerance $\Delta\theta>0$, we draw a candidate $\eta^\dagger\sim\pi_\eta$ and accept it if $\theta-\Delta\theta\leq|\eta^\dagger|^2\leq\theta+\Delta\theta$. If the candidate is rejected, we draw a new $\eta^\dagger\sim\pi_\eta$ and repeat the procedure.
Here, we take $\Delta \theta= 1$ for an acceptance rate of $\sim 1\%$.
The results are presented on Figure \ref{fig:LinearCase_GO}.
We observe that, compared to the conservative bound (which again simplifies to the standard setting of Figure \ref{fig:LinearCase_a}), the incremental one positions observation locations differently and yields a slightly more concentrated posterior density which better represents the true value of $\theta$.

\begin{figure}
    \centering
    \begin{subfigure}{0.48\textwidth}
        \includegraphics[width=\linewidth]{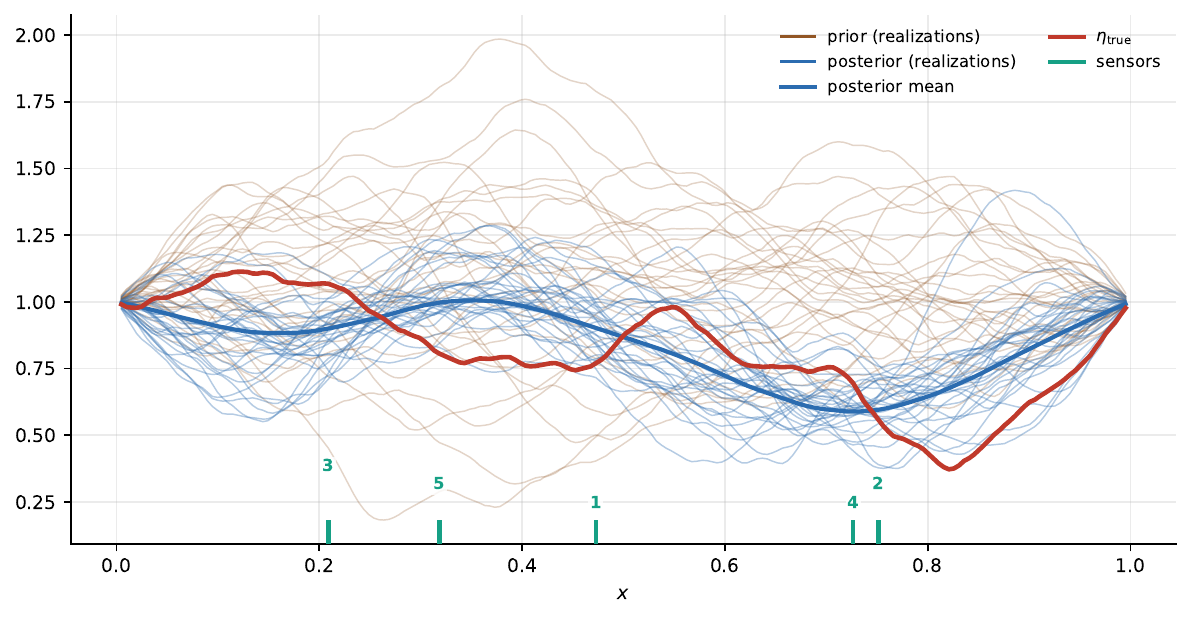}
        \caption{Prior and posterior realizations of the initial condition $\eta$ using \texttt{i-SNR} designs}
        \label{fig:energy_reconstruction}
    \end{subfigure}
    \begin{subfigure}{0.45\textwidth}
        \includegraphics[width=\linewidth]{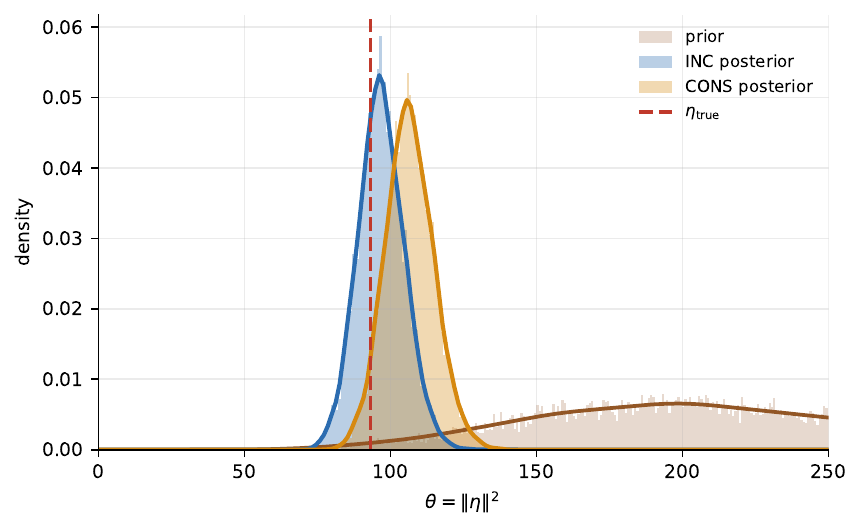}
        \caption{Histograms of the prior and posterior of $\theta$ using \texttt{i-SNR} and \texttt{c-SNR} designs.}
        \label{fig:energy_histogram}
    \end{subfigure}

    \caption{Linear forward model $\lambda=0$, goal-oriented setting $\theta=\|\eta\|^2$. Left: prior and posterior samples of the initial condition $\eta$ using a design of $m=5$ experiments obtained by the greedy algorithm \eqref{eq:Greedy} to optimize the incremental SNR (\texttt{i-SNR}) $W_m\mapsto \frac{1}{2}\ln(\frac{|W_m^\top \Cov(Y) W_m|}{|W_m^\top \Sigma_{Y|\theta} W_m|})$ with $\Sigma_{Y|\theta}$ estimated using \eqref{eq:estimator_SigmaYTheta}. Right: comparison of the posteriors for $\theta$ obtained by either the \texttt{i-SNR} design or the \texttt{c-SNR} design.}
    \label{fig:LinearCase_GO}
\end{figure}

\subsection{SNRs as certified bounds for the EIG}\label{sec:certification}
In this section, the certified bounds for the EIG are assessed as follows. For a given setting (here, standard or auxiliary) and a given $\lambda\in[0,1]$, we employ the greedy algorithm \eqref{eq:Greedy} to identify a design $W_m$ with $m$ observations using either the incremental SNR (\texttt{i-SNR}, \eqref{eq:incremental_optimization}) or the conservative SNR (\texttt{c-SNR}, \eqref{eq:conservative_optimization}), where $\Sigma_\mathrm{signal}$ and $\Sigma_\mathrm{noise}$ are computed either with the gradient-based approach (\texttt{gradient}, \Cref{prop:BoundFisher}) or with the approximation-based approach (\texttt{affine+nn}: affine plus neural network correction, as in \Cref{prop:Sigma_FG}), remember Table \ref{tab:formula}. {The neural network we use consists of three fully connected layers of width $2p$, each with a GELU activation function.}
Once $W_m$ is computed, we plot the four bounds \texttt{inc\_lb}, \texttt{inc\_ub}, \texttt{cons\_lb},
\texttt{cons\_ub} of $\mathrm{EIG}(W_m)$ given by Corollaries \ref{cor:Bound_EIG_incremental} and \ref{cor:Bound_EIG_conservative} as function of $m$. For the conservative bounds, the quantity $\mathrm{EIG}(I_p)$ is computed using a nested Monte Carlo estimator with 1000 samples at each level (hence, a total of $1000^2$ samples).

\begin{figure}
	\centering
	\begin{subfigure}[t]{0.48\textwidth}
		\includegraphics[width=\linewidth]{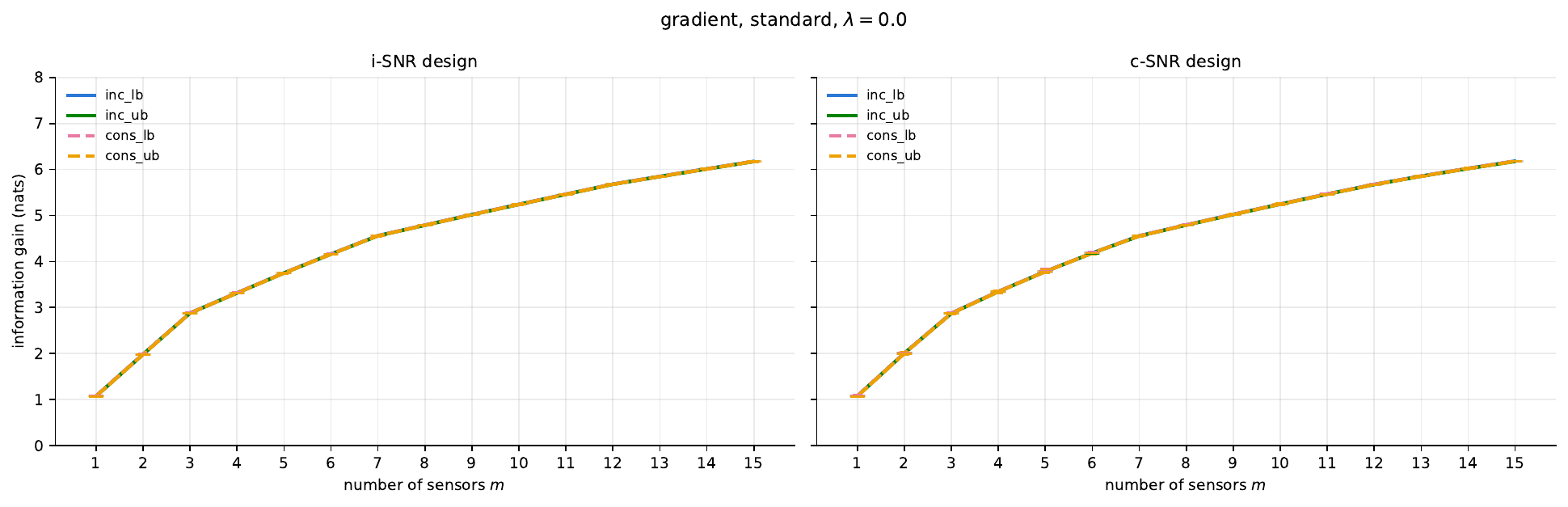}
		\caption{$\lambda = 0$}
	\end{subfigure}
	\hfill
	\begin{subfigure}[t]{0.48\textwidth}
		\includegraphics[width=\linewidth]{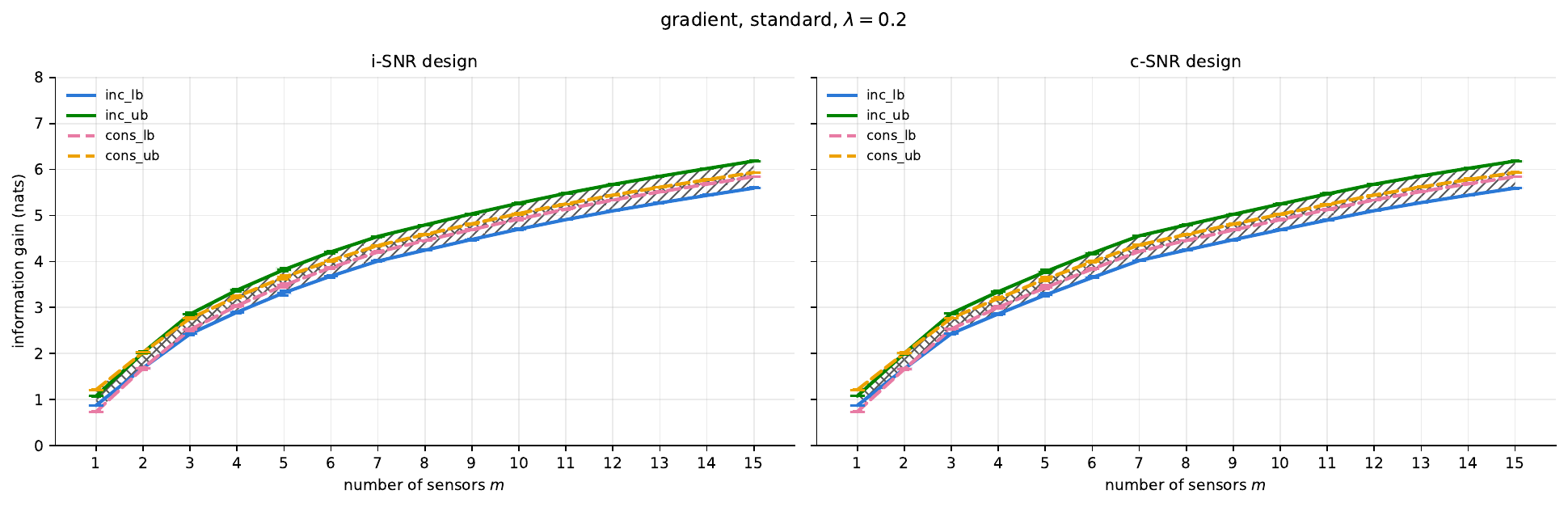}
		\caption{$\lambda = 0.2$}
	\end{subfigure}

	\vspace{1em}

	\begin{subfigure}[t]{0.48\textwidth}
		\includegraphics[width=\linewidth]{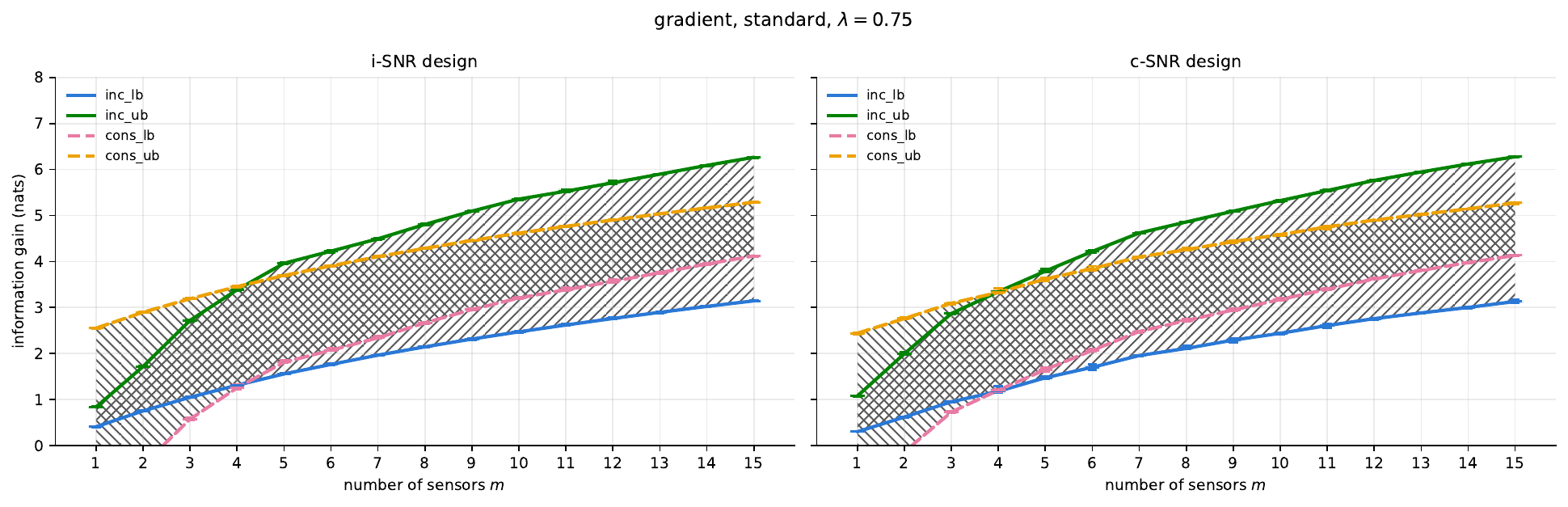}
		\caption{$\lambda = 0.75$}
	\end{subfigure}
	\hfill
	\begin{subfigure}[t]{0.48\textwidth}
		\includegraphics[width=\linewidth]{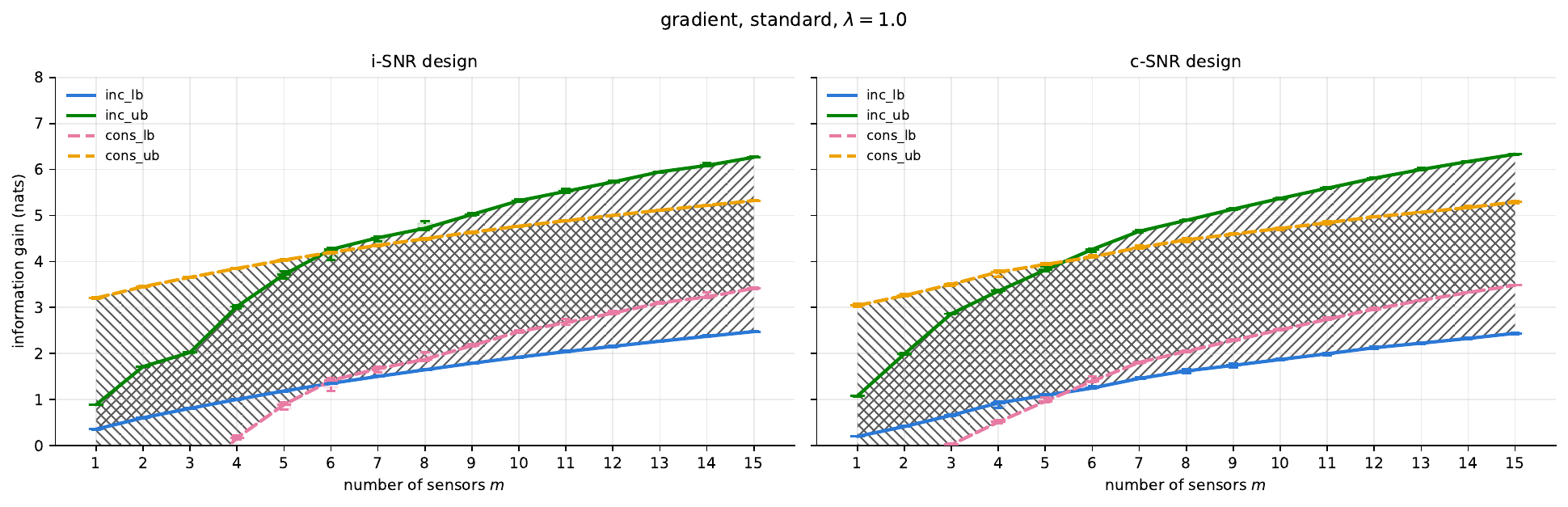}
		\caption{$\lambda = 1$}
		\label{fig:bounds-gradient-standard_d}
	\end{subfigure}
	\caption{%
	Standard setting \& \texttt{gradient}-based diagnostics: certified bounds (\texttt{inc\_lb}, \texttt{inc\_ub}, \texttt{cons\_lb}, \texttt{cons\_ub}) versus number of observations $m$, under incremental (\texttt{i-SNR}, left) and conservative (\texttt{c-SNR}, right) designs, for different $\lambda$.
	}
	\label{fig:bounds-gradient-standard}
\end{figure}

Figure \ref{fig:bounds-gradient-standard} represents the standard setting for different $\lambda$, where we use the \texttt{gradient} approach to compute $\Sigma_\mathrm{signal}$ and $\Sigma_\mathrm{noise}$.
For $\lambda = 0$, the four estimates are numerically indistinguishable, which is consistent with Remark \ref{rmk:GaussianY} since $(Y,\theta)$ is jointly Gaussian.
The intervals then grow with $\lambda$, with the conservative interval growing faster than the incremental one.
At strong nonlinearity ($\lambda = 1$,\Cref{fig:bounds-gradient-standard_d}) the conservative estimates bracket a wider interval; \texttt{cons\_lb} and \texttt{cons\_ub} span almost the full plotted range at small budgets, while the incremental estimates \texttt{inc\_lb} and \texttt{inc\_ub} remain comparatively tight.

\begin{figure}
	\centering
	\begin{subfigure}[t]{0.48\textwidth}
		\includegraphics[width=\linewidth]{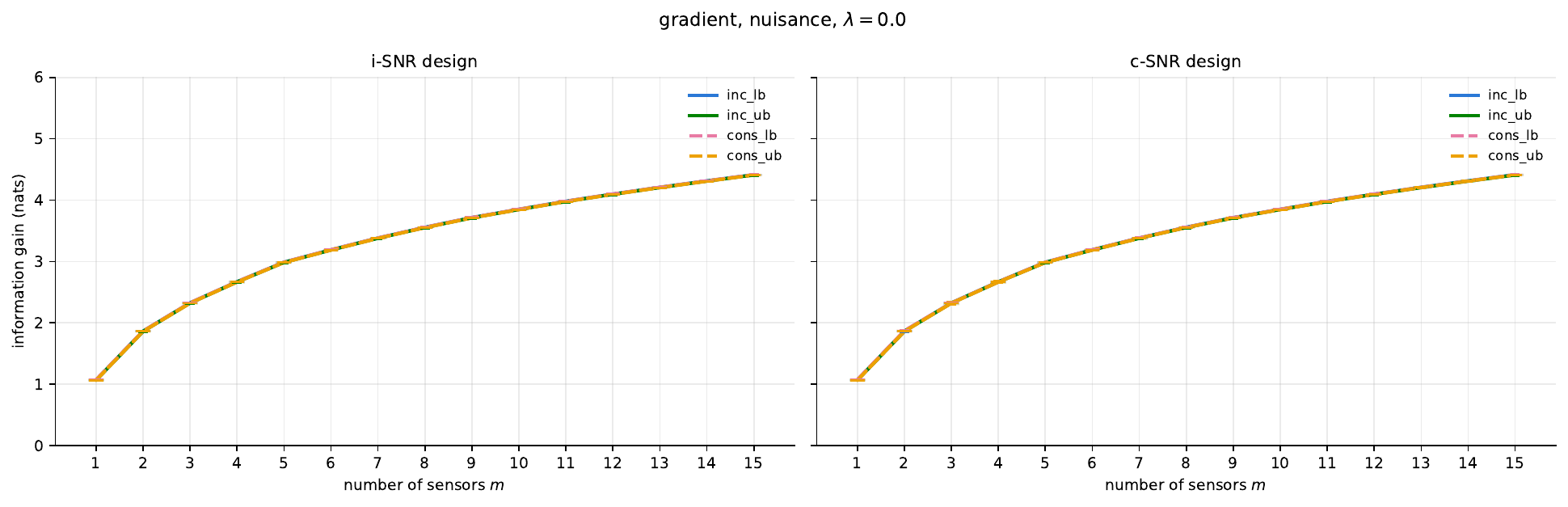}
		\caption{$\lambda = 0$}
	\end{subfigure}
	\hfill
	\begin{subfigure}[t]{0.48\textwidth}
		\includegraphics[width=\linewidth]{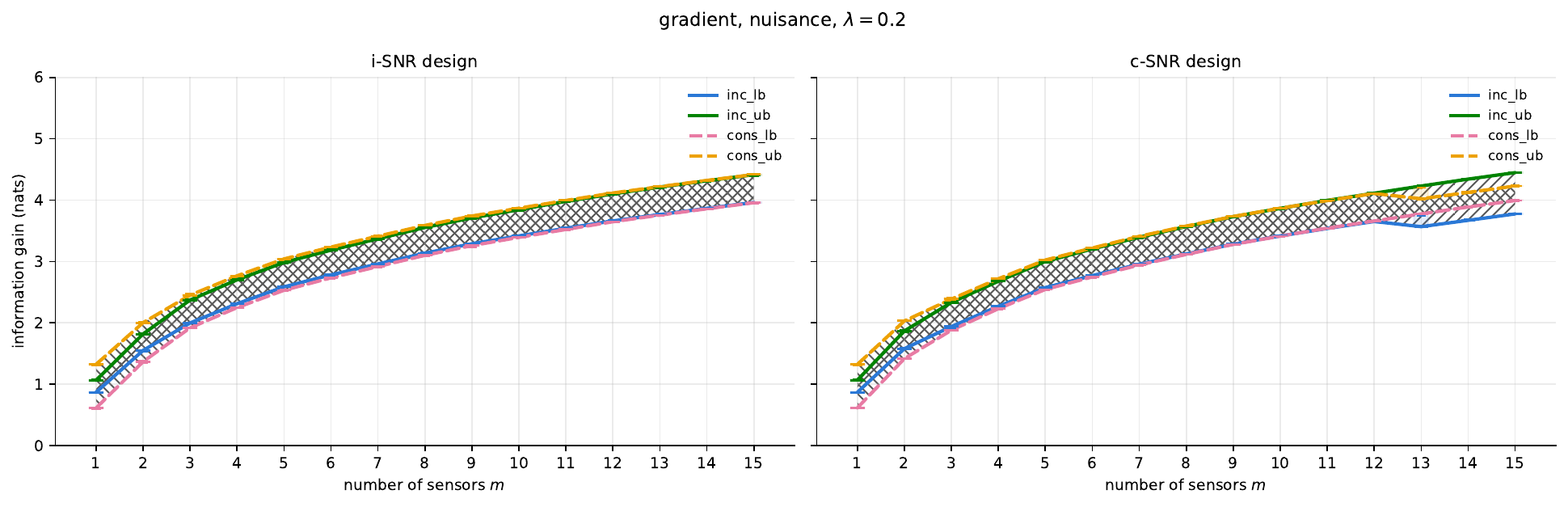}
		\caption{$\lambda = 0.2$}
	\end{subfigure}

	\vspace{1em}

	\begin{subfigure}[t]{0.48\textwidth}
		\includegraphics[width=\linewidth]{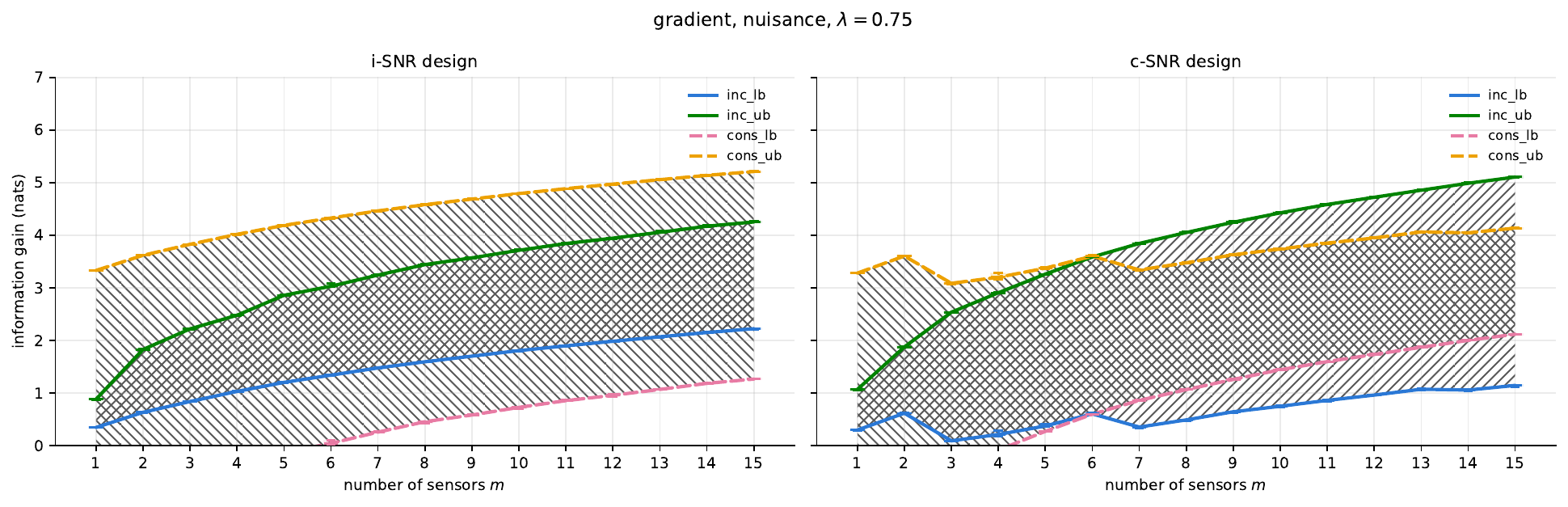}
		\caption{$\lambda = 0.75$}
		\label{fig:bounds-gradient-go_075}
	\end{subfigure}
	\hfill
	\begin{subfigure}[t]{0.48\textwidth}
		\includegraphics[width=\linewidth]{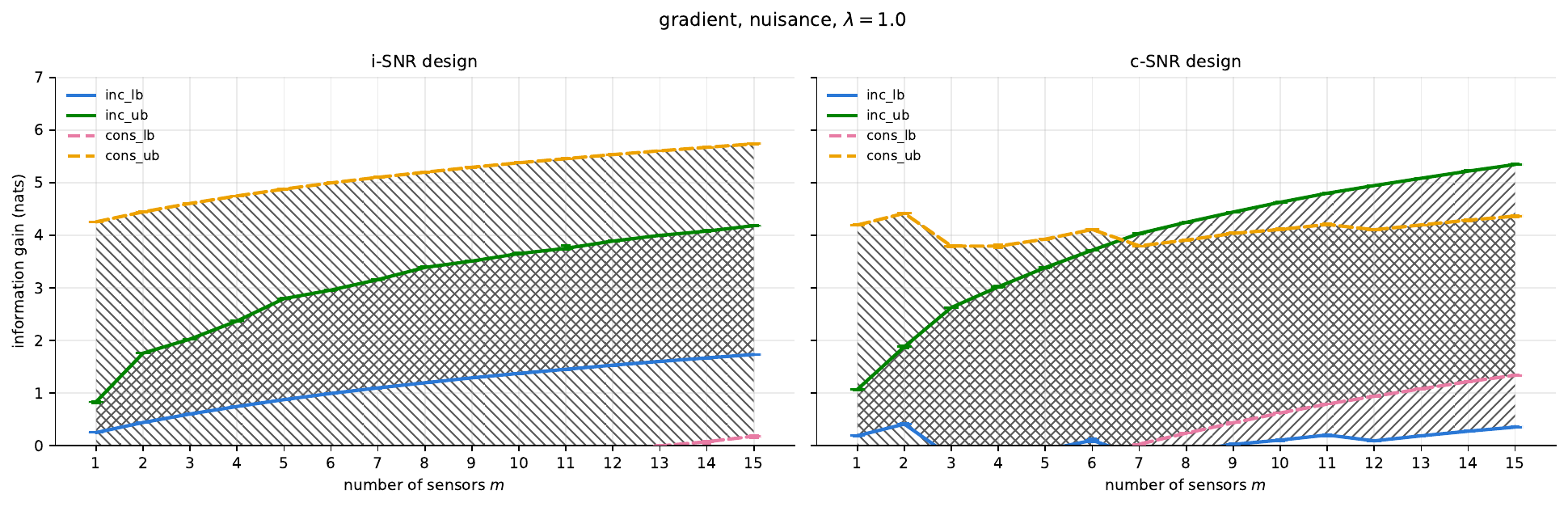}
		\caption{$\lambda = 1$}
	\end{subfigure}
	\caption{%
		Auxiliary setting \& \texttt{gradient}-based diagnostics: certified bounds (\texttt{inc\_lb}, \texttt{inc\_ub}, \texttt{cons\_lb}, \texttt{cons\_ub}) versus number of observations $m$, under incremental (\texttt{i-SNR}, left) and conservative (\texttt{c-SNR}, right) SNR designs, for different $\lambda$.
	}
	\label{fig:bounds-gradient-go}
\end{figure}

\Cref{fig:bounds-gradient-go} corresponds to the auxiliary setting, with the diagnostics computed via the \texttt{gradient} approach.
The behavior is similar to that observed in the standard case (\Cref{fig:bounds-gradient-standard}): the intervals collapse when $\lambda = 0$ and widen as $\lambda$ increases. Moreover, the incremental intervals are relatively narrow for small values of $m$ and widen as $m$ increases, while the opposite trend is observed for the conservative bounds.
It is worth noting on Figure \ref{fig:bounds-gradient-go_075} that the lower-bound \texttt{inc\_lb} is not monotonically increasing with $m$ when the design $W_m$ is identified by maximizing the \texttt{c-SNR}: maximizing one objective (the conservative lower bound) doesn't ensure being optimal for another objective (the incremental lower bound).

\begin{figure}
	\centering
	\begin{subfigure}[t]{0.48\textwidth}
		\includegraphics[width=\linewidth]{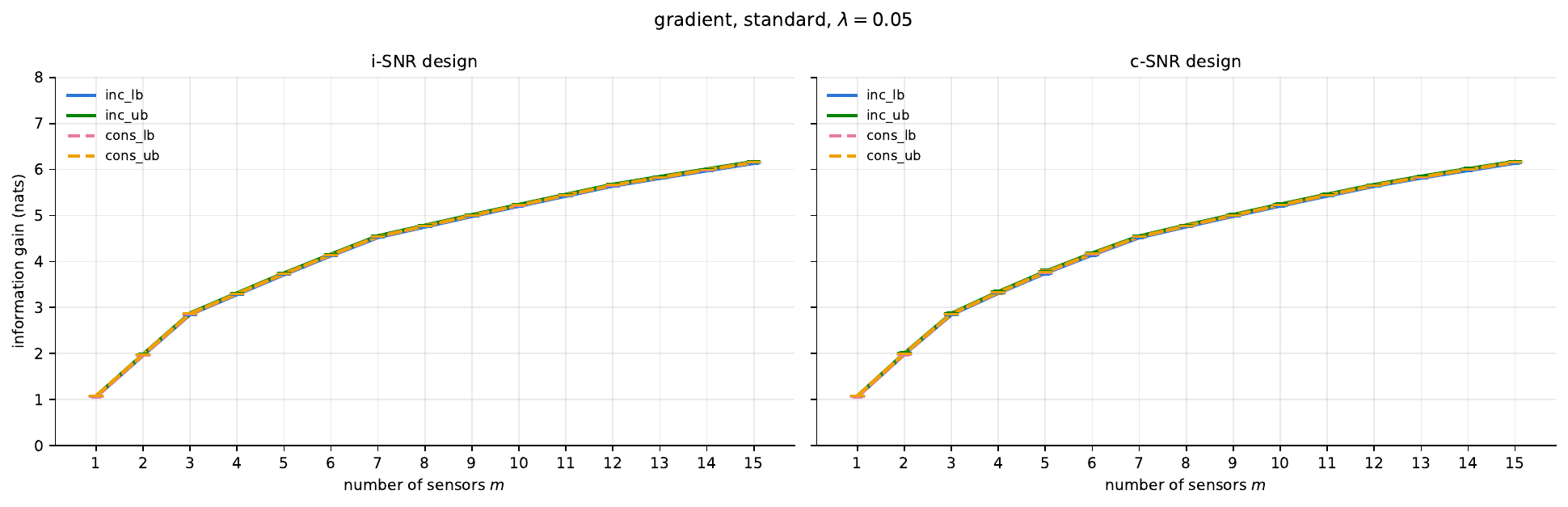}
		\caption{\texttt{gradient}, $\lambda = 0.05$}
	\end{subfigure}
	\hfill
	\begin{subfigure}[t]{0.48\textwidth}
		\includegraphics[width=\linewidth]{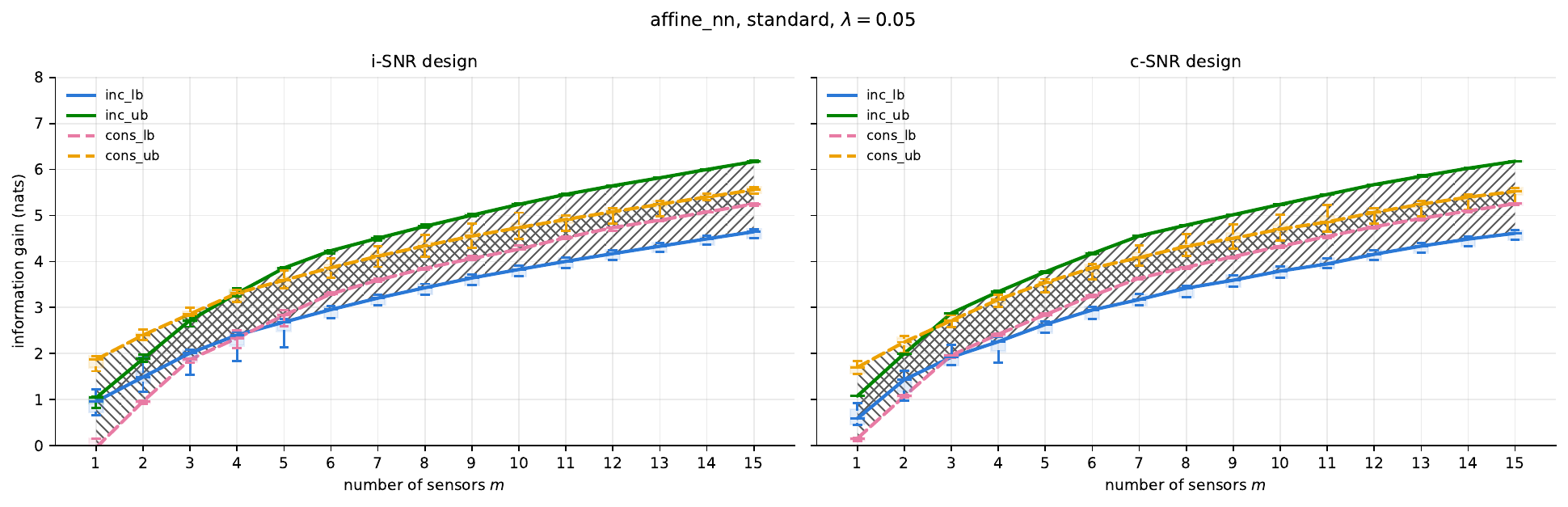}
		\caption{\texttt{affine+nn}, $\lambda = 0.05$}
	\end{subfigure}

	\vspace{1em}

	\begin{subfigure}[t]{0.48\textwidth}
		\includegraphics[width=\linewidth]{images/03_boxplot_gradient_standard_lambda_1.00.pdf}
		\caption{\texttt{gradient}, $\lambda = 1$}
	\end{subfigure}
	\hfill
	\begin{subfigure}[t]{0.48\textwidth}
		\includegraphics[width=\linewidth]{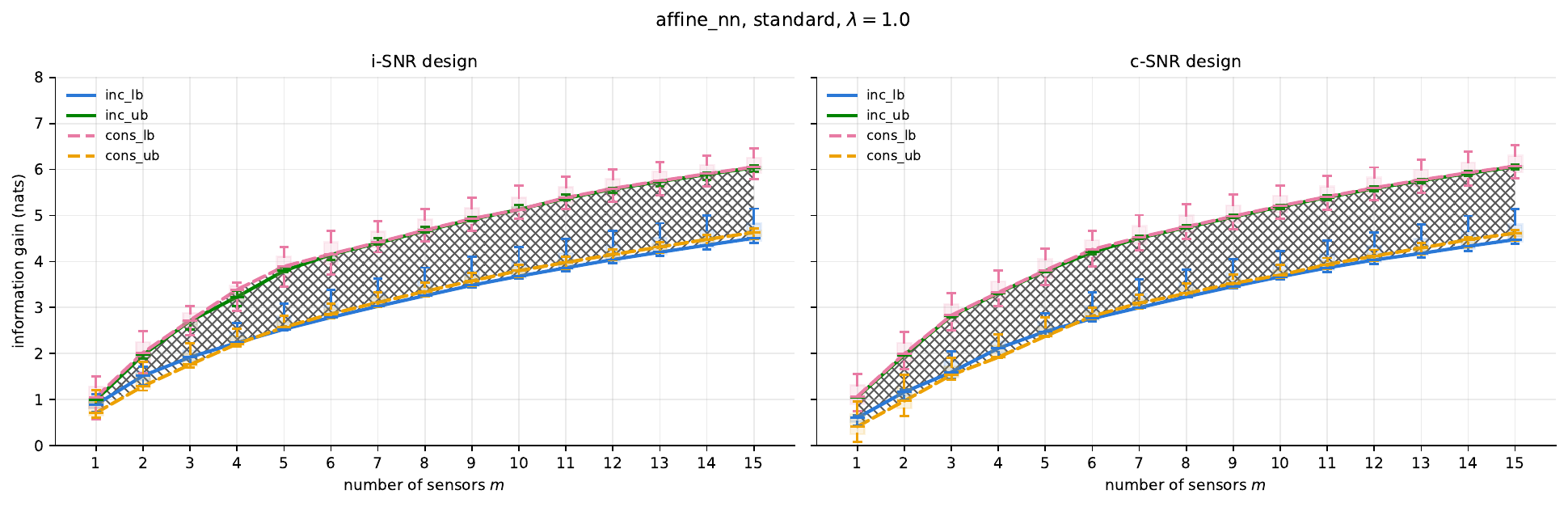}
		\caption{\texttt{affine+nn}, $\lambda = 1$}
	\end{subfigure}
	\caption{%
        Standard setting: comparison of the \texttt{gradient} approach (left) and the \texttt{affine+nn} approach (right) for computing the diagnostics $\Sigma_\mathrm{noise}$ and $\Sigma_\mathrm{signal}$ at weak ($\lambda = 0.05$, top) and strong ($\lambda = 1$, bottom)
		nonlinearity.
	}
	\label{fig:cost-of-certification}
\end{figure}

Figure \ref{fig:cost-of-certification} compares the \texttt{gradient} and the \texttt{affine+nn} approach for computing the diagnostics $\Sigma_\mathrm{noise}$ and $\Sigma_\mathrm{signal}$.
We observe that the intervals obtained with the \texttt{gradient} approach are tighter than those obtained with \texttt{affine+nn} for small values of the nonlinearity parameter $\lambda$. However, the \texttt{gradient}-based intervals widen much faster as $\lambda$ increases: for large values of $\lambda$, the \texttt{affine+nn} approach provides a more robust alternative.
This comparison also highlights the \emph{cost of certification}. Unlike the \texttt{gradient} approach, which provides certified bounds (remember Proposition \ref{prop:BoundFisher}), the \texttt{affine+nn} approach does not yield certified intervals unless the approximation error is neglected, as discussed at the end of Section \ref{sec:Approxiamtion_based_diagnostics}. Thus, while \texttt{affine+nn} can provide tighter intervals, this comes at the cost of losing the certification of the EIG.

\section{Conclusion}
In this work, we developed tractable upper and lower bounds for the expected information gain (EIG) that take the form of signal-to-noise ratios (SNR). These bounds coincide with the true EIG in the linear \& Gaussian setting, and we show that the gap between the bounds is directly related to the nonGaussianity of the problem. Since the bounds depend only on four diagnostic matrices, involving covariance and Fisher information, they provide a promising route toward scalable experimental design for high-dimensional parameter and observable spaces.

We proposed two approaches for computing these diagnostic matrices. The first is a gradient-based approach, which preserves the certification of the bounds. The second is an approximation-based approach, which is more robust in strongly nonlinear regimes but does so at the cost of losing certification when the approximation error is not accounted for. This leads to a practical guideline: when gradients are available and certified bounds are required, the gradient-based approach is appropriate; when gradients are unavailable and certification is not essential, the approximation-based approach provides an interesting alternative.

Several directions remain open for future work. First, an important challenge is to develop efficient ways of computing the diagnostic matrices for very large-scale systems, with, for example, $p=\mathcal{O}(10^6)$ observables. Such applications would also provide an opportunity to better characterize the differences between the conservative and incremental SNR bounds, particularly in regimes where these differences are expected to be substantial.
Second, reversing the roles of $\theta$ and $Y$ may provide a complementary perspective for identifying the parameters that are most informative given an observed dataset.
Finally, transport-based techniques could be used to precondition the joint distribution of $(Y,\theta)$, with the aim of making it more Gaussian and thereby improving the tightness of the proposed bounds.

\section*{Acknowledgments}
The authors would like to thank Matt Li and Mathieu Le Provost for helpful discussions on this work.
As part of the "France 2030" initiative, this work has benefited from a national grant managed by the French National Research Agency (Agence Nationale de la Recherche) attributed to the Exa-MA project of the NumPEx PEPR program, under the reference ANR-22-EXNU-0002.

\appendix

\section{Basic properties of Fisher information matrices}\label{sec:Fisher}
For simplicity, we will use the notation $(\cdot)^{\otimes 2}$ for the outer product so that $ v^{\otimes 2} := v v^\top \in\R^{p\times p}$ for any vector $v\in\R^p$.
Let $X$ be a random vector in $\R^d$ with smooth probability density function $\pi_X$ such that $\mathrm{supp}(\pi_X)=\R^d$ and with a score $\nabla\ln\pi_X(X)$ square integrable. The Fisher information of $X$ is defined as
$$
\mathcal{I}_{X} := \E[ ( \nabla\ln\pi_{X}(X) )^{\otimes 2} ]
= \Cov( \nabla\ln\pi_{X}(X) ) ,
$$
where we used the fact that $\E[ \nabla\ln\pi_{X}(X)  ] = \int \nabla \pi_{X}(x) \d x = 0$.
One integration by parts yields $\E[(X-m)\nabla\ln\pi_{X}(X)^\top] = \int I_d \pi_X(x)\d x = I_d $, where $m=\E[X]$. Then, a Cauchy-Schwarz inequality yields the Cramér-Rao bound
$$
  (\mathcal{I}_X)^{-1}   \preceq \Cov(X)  .
$$
A well known result is that equality in the above relation implies that $\ln\pi_{X}(X)$ and $X-m$ are collinear, therefore $X$ is necessarily Gaussian:
$$
  I_d   = \Cov(X) \mathcal{I}_X
  \qquad\Leftrightarrow\qquad
  X = \mathcal{N}( m,\Cov(X) )
$$

Consider now a pair of random vectors $(X,Y)$ in $\R^d\times \R^p$ with smooth probability density function $\pi_{X,Y}$. The Fisher information matrix of $(X,Y)$ admits the block decomposition
\begin{align*}
 \mathcal{I}_{(X,Y)}
 &= \E\left[
 \begin{pmatrix}
  \nabla_X\ln\pi_{X,Y}(X,Y) \\
  \nabla_Y\ln\pi_{X,Y}(X,Y) \\
 \end{pmatrix} ^{\otimes 2} \right]
 =\begin{pmatrix}
   \E[\mathcal{I}_{X|Y}]  & \mathcal{I}_{X , Y} \\
   \mathcal{I}_{Y , X}  & \E[\mathcal{I}_{Y|X}] \\
  \end{pmatrix} ,
\end{align*}
where
\begin{align*}
 \mathcal{I}_{X , Y}
 &= \E[ ( \nabla_X\ln\pi_{X,Y}(X,Y) ) ( \nabla_Y\ln\pi_{X,Y}(X,Y) ) ^\top ] ,
\end{align*}
is the cross-information matrix and
\begin{align*}
 \E[\mathcal{I}_{Y|X}]
 &=\E\Big[ \E[( \nabla_Y\ln\pi_{X|Y}(X|Y) )^{\otimes 2}|X] \Big]
 = \E[ ( \nabla_Y\ln\pi_{X,Y}(X,Y) )^{\otimes 2} ] ,
\end{align*}
is the expected Fisher information matrix, that is, the expectation over $Y$ of the (random) Fisher information matrix $\mathcal{I}_{X|Y}$ of $X$ conditioned on $Y$.
The following key proposition shows that the Fisher information of a marginal density is upper-bounded by the Shur complement of the Fisher information of the joint density.

\begin{proposition}\label{prop:Fisher}
 Let $(X,Y,Z)$ a random vector in $\R^{d}\times \R^{p} \times \R^{q}$.
 Then
 \begin{align}
  \mathcal{I}_{Y} &=  \E[\mathcal{I}_{Y|X}]  - \E[\Cov(\nabla_{Y} \ln\pi_{X,Y}(X,Y) | Y)] \label{eq:FisherMarginal} \\
  \E[\mathcal{I}_{Y|X}] &= \E[\mathcal{I}_{Y|X,Z}]  -  \E[\Cov(\nabla_{Y} \ln\pi_{X,Y,Z}(X,Y,Z) | X,Y)] , \label{eq:FisherMarginal_bis}
 \end{align}
 and
\begin{align}
  \mathcal{I}_{Y}&\preceq  \E[\mathcal{I}_{Y|X}]  - \mathcal{I}_{Y,X}   \E[\mathcal{I}_{X|Y}] ^{-1} \mathcal{I}_{X,Y}  \label{eq:FisherMarginal_bound}\\
 \E[\mathcal{I}_{Y|X}]  &\preceq \E[\mathcal{I}_{Y|X,Z}] -  \E[\mathcal{I}_{Y,Z|X}] \E[\mathcal{I}_{Z|X,Y}]^{-1} \E[\mathcal{I}_{Z,Y|X}] .
 \label{eq:FisherMarginal_bound_bis}
\end{align}

\end{proposition}

\begin{proof}
 Because the marginal density writes $\pi_Y(y)=\int \pi_{X,Y}(x,y)\d x$, we have
 \begin{align*}
  \nabla_Y\ln\pi_Y(y)
  &= \frac{\int \nabla_Y\pi_{X,Y}(x,y)\d x }{\pi_Y(y)}
  = \E[\nabla_Y\ln\pi_{X,Y}(X,Y)|Y=y] ,
 \end{align*}
 and then $\mathcal{I}_Y=\Cov(\nabla_Y\ln\pi_Y(Y)) = \Cov( \E[\nabla_Y\ln\pi_{X,Y}(X,Y)|Y] )$.
 Then, using the law of total variance, we have
 \begin{align*}
  \E[\mathcal{I}_{Y|X}]
  &=  \Cov( \nabla_Y\ln\pi_{X,Y}(X,Y) )\\
  &= \Cov( \E[ \nabla_Y\ln\pi_{X,Y}(X,Y) |Y] )  + \E[\Cov( \nabla_Y\ln\pi_{X,Y}(X,Y) |Y )] \\
  &=  \mathcal{I}_Y + \E[\Cov( \nabla_Y\ln\pi_{X,Y}(X,Y) |Y )] ,
 \end{align*}
 which is \eqref{eq:FisherMarginal}. Relation \eqref{eq:FisherMarginal_bis} is derived using the same calculation.
 To show \eqref{eq:FisherMarginal_bound}, let us note that for any matrix $A\in\R^{d\times p}$ we have
 \begin{align*}
  \mathcal{I}_Y
  &= \Cov\Big( \E[\nabla_Y\ln\pi_{X,Y}(X,Y)|Y] \Big) \\
  &= \Cov\Big( \E[\nabla_Y\ln\pi_{X,Y}(X,Y)|Y] - A \underbrace{\E[\nabla_X\ln\pi_{X,Y}(X,Y)|Y]}_{=0} \Big) \\
  &= \Cov\Big( \E\Big[\nabla_Y\ln\pi_{X,Y}(X,Y) - A  \nabla_X\ln\pi_{X,Y}(X,Y) \Big|Y\Big] \Big) \\
  &\preceq \Cov\Big(  \nabla_Y\ln\pi_{X,Y}(X,Y) - A  \nabla_X\ln\pi_{X,Y}(X,Y)  \Big) \\
  &= \E[\mathcal{I}_{Y|X}] - \mathcal{I}_{Y , X}A^\top - A \mathcal{I}_{X , Y} +  A \E[\mathcal{I}_{X|Y}] A^\top  .
 \end{align*}
 Taking $A = \mathcal{I}_{Y,X} \E[\mathcal{I}_{X|Y}]^{-1}  $ gives \eqref{eq:FisherMarginal_bound}.
 Relation \eqref{eq:FisherMarginal_bound_bis} is derived using the same calculation.
 This concludes the proof.
\end{proof}

\section{Dimensional logarithmic Sobolev inequalities}\label{sec:dLSI}
Let $\mu_0=\mathcal{N}( 0,I_d )$ be the standard Gaussian measure on $\R^d$ and let $\nu$ be any smooth density on $\R^d$. Following \cite{bakry2012dimension}, the dimensional logarithmic Sobolev inequalities (LSI) read
\begin{align*}
 \Dkl(\nu||\mu_0) &\leq \frac{1}{2} \int \| x \|^2 \d\nu - \frac{d}{2} + \frac{1}{2}\ln\left|\int (\nabla\ln \nu) ^{\otimes 2} \d\nu \right| \\
 \Dkl(\nu||\mu_0) &\geq \frac{1}{2} \int \| x \|^2 \d\nu - \frac{d}{2} - \frac{1}{2}\ln\left| \int x^{\otimes 2} \d\nu - \left( \int x \d\nu \right)^{\otimes 2} \right| .
\end{align*}
We refer to \cite[Appendix C.2]{li2025sharp} for a direct proof that the standard (reverse) Gaussian logarithmic Sobolev inequality implies the above dimensional LSI.
These inequalities allow us to derive the following key result.

\begin{proposition}\label{prop:TheResult}
 Let $(\theta,Y,Y')$ be a random vector on $\R^d\times \R^m\times \R^{m'}$ with smooth joint probability density function $\pi_{\theta,Y,Y'}$.
 Then
 \begin{equation}\label{eq:TheResult}
 \frac{1}{2}    \ln\left(\frac{| \E[\mathcal{I}_{Y' | Y}]^{-1} |}{|\E[\Cov(Y'| \theta,Y)]|}\right)
 \leq
 \E_Y[\Dkl( \pi_{\theta|Y,Y'} || \pi_{\theta|Y} )]
 \leq
 \frac{1}{2}   \ln\left( \frac{| \E[\Cov(Y'| Y)] |}{|\E[\mathcal{I}_{Y' | \theta,Y}]^{-1}|} \right) ,
 \end{equation}
 where $\E[\mathcal{I}_{Y' | Y}] = \E[ (\nabla_{Y'} \ln\pi_{Y,Y'}(Y,Y') )^{\otimes 2} ]$ and $\E[\mathcal{I}_{Y' |\theta,Y}] = \E[ (\nabla_{Y'} \ln\pi_{\theta,Y,Y'}(\theta,Y,Y') )^{\otimes 2} ]$ are the expected Fisher information matrix of $(Y'|Y)$ and $(Y'|\theta,Y)$, respectively.

\end{proposition}

To prove Proposition \ref{prop:TheResult}, we will need the following lemma.

\begin{lemma}\label{lem:dLSI_Fisher}
 Let $(X,Z)$ be a random vector in $\R^{d} \times \R^{p}$ with smooth joint density $\pi_{X,Z}$. Assume $\E[\|X\|^2]<\infty$ and let  $ \widetilde\pi_{X} = \mathcal{N}( \E[X] , \Cov(X) ) $. Then
\begin{align}
 \frac{1}{2}   \ln\left(\frac{|\Cov(X)|}{|\E[\Cov(X| Z)]|}\right)
 \leq \E_{Z}[\Dkl( \pi_{X|Z} || \widetilde\pi_{X})]
 &\leq \frac{1}{2}   \ln\left( \left| \Cov(X) \E[\mathcal{I}_{X|Z}] \right| \right) , \label{eq:dLSI_Fisher_Ytheta}
\end{align}
where $\E[\Cov(X|Z)] = \E[( X - \E[X|Z] )^{\otimes 2}] $ is the expected conditional covariance of $X | Z$ and $\E[\mathcal{I}_{X|Z}] = \E[ (\nabla_{X} \ln\pi_{X,Z}(X,Z) )^{\otimes 2} ]$ is the expected Fisher information matrix of $X|Z$.
\end{lemma}

We mention here that the left inequality in \eqref{eq:dLSI_Fisher_Ytheta} is related to the upper bound on differential entropy given in \cite[Theorem 9.6.5]{cover1991elements}, whereas the right inequality in \eqref{eq:dLSI_Fisher_Ytheta} relates to Efroimovich inequality \cite{efroimovich1980information} (see also \cite{courtadeinformation} for a link with Bayesian Cramér--Rao Bound), which gives a lower bound on differential entropy.

\begin{proof}[Proof of Lemma \ref{lem:dLSI_Fisher}]
 First we note that the dimensional logarithmic Sobolev inequalities (dLSI) can be naturally extended to any Gaussian measure $\mu=\mathcal{N}(m,\Sigma)$ with full rank covariance matrix $\Sigma\in\R^{d\times d}$ as follows. Let $T(x) =  \Sigma^{-1/2}(x-m)$ be the affine transformation such that $T_\sharp\mu = \mu_0$. We have $\Dkl(\nu||\mu) = \Dkl( T_\sharp\nu ||\mu_0)$ so that the dLSI yields
 \begin{equation}
  \Dkl(\nu||\mu) \leq \frac{1}{2} \int \| x-m \|_{\Sigma^{-1}}^2 \d\nu - \frac{d}{2} + \frac{1}{2}\ln\left| \Sigma \int  (\nabla\ln\nu)^{\otimes 2} \d\nu \right| \label{eq:dLSI_up}
 \end{equation}
and
\begin{align}
 \Dkl(\nu||\mu) &\geq \frac{1}{2} \int \| x-m \|_{\Sigma^{-1}}^2 \d\nu - \frac{d}{2} - \frac{1}{2}\ln\frac{\left| \int ( x-m)^{\otimes 2} \d\nu - \left( \int (x-m) \d\nu \right)^{\otimes 2} \right| }{|\Sigma|}  \nonumber\\
 &= \frac{1}{2} \int \| x-m \|_{\Sigma^{-1}}^2 \d\nu - \frac{d}{2} - \frac{1}{2}\ln\frac{\left| \int x^{\otimes 2} \d\nu - \left( \int x \d\nu \right)^{\otimes 2} \right| }{|\Sigma|}  \label{eq:dLSI_down}.
\end{align}
 Thus we can write
 \begin{align*}
 \E_{Z}[\Dkl( \pi_{X|Z} || \widetilde\pi_{X})]
 &\overset{\eqref{eq:dLSI_up}}{\leq} \E_{Z}\Bigg[  \frac{1}{2} \int \| y_\perp-\E[X] \|_{\Cov(X)^{-1}}^2 \d\pi_{X|Z} - \frac{d}{2} \\
 &\qquad\qquad + \frac{1}{2}\ln\left(\left| \Cov(X) \int  (\nabla_{X}\ln\pi_{X|Z})^{\otimes 2} \d\pi_{X|Z} \right|\right)  \Bigg] \\
 &= \frac{1}{2}  \E_{Z}\left[ \ln\left(\left| \Cov(X) \int  (\nabla_{X}\ln\pi_{X|Z})^{\otimes 2} \d\pi_{X|Z} \right|\right)  \right] \\
 &\leq \frac{1}{2}   \ln\left(\left| \Cov(X) \int  (\nabla_{X}\ln\pi_{X|Z})^{\otimes 2} \d\pi_{X,Z} \right|\right) ,
\end{align*}
where in the last step we used the Jensen inequality $\E[\varphi(X)]\leq \varphi(\E[X])$ with $\varphi(X)=\ln|X|$ a concave function on the space of SPD matrices. This gives the right-hand side of \eqref{eq:dLSI_Fisher_Ytheta}. To prove the left-hand side, we write
 \begin{align*}
 \E_{Z}\left[ \Dkl( \pi_{X|Z} || \widetilde\pi_{X} )\right]
 &\overset{\eqref{eq:dLSI_down}}{\geq} \E_{Z}\Bigg[  \frac{1}{2} \int \| y_\perp-\E[X] \|_{\Cov(X)^{-1}}^2 \d\pi_{X|Z} - \frac{d}{2} \\
 &\qquad\qquad - \frac{1}{2}\ln\left(\frac{\left| \int  y_\perp^{\otimes 2} \d\pi_{X|Z} - \left( \int y_\perp \d\pi_{X|Z} \right)^{\otimes 2} \right| }{|\Cov(X)|} \right)  \Bigg] \\
 &= - \frac{1}{2} \E_{Z}\Bigg[\ln\left(\frac{\left| \int  y_\perp^{\otimes 2} \d\pi_{X|Z} - \left( \int y_\perp \d\pi_{X|Z} \right)^{\otimes 2} \right| }{|\Cov(X)|} \right)  \Bigg] \\
 &\geq - \frac{1}{2} \ln\left(\frac{\left| \int  y_\perp^{\otimes 2} \d\pi_{X} - \E_{Z}[( \int y_\perp \d\pi_{X|Z} )^{\otimes 2} ]\right| }{|\Cov(X)|} \right) ,
\end{align*}
where we again used the Jensen inequality in the last step. Because $\E[\Cov(X|Z)]=\E[X^{\otimes2}]-\E[\E[X|Z]^{\otimes2}]$, this yields the left-hand side of \eqref{eq:dLSI_Fisher_Ytheta} and concludes the proof.
\end{proof}

We are now able to prove Proposition \ref{prop:TheResult}.

\begin{proof}[Proof of Proposition \ref{prop:TheResult}]
 We can write
 \begin{align*}
 \E_Y[\Dkl( \pi_{\theta|Y,Y'} || \pi_{\theta|Y} )]
 &= \int \ln\left(\frac{\pi_{\theta|Y,Y'}}{\pi_{\theta|Y}}\right) \d\pi_{\theta,Y,Y'} \\
 &= \int \ln\left(\frac{\pi_{Y'|\theta,Y}}{\pi_{Y'|Y}}\right) \d\pi_{\theta,Y,Y'} \\
 &= \int \ln\left(\frac{\pi_{Y'|\theta,Y}}{\widetilde\pi_{Y'}}\right) \d\pi_{\theta,Y,Y'}  - \int \ln\left(\frac{\pi_{Y'|Y}}{\widetilde\pi_{Y'}}\right) \d\pi_{Y,Y'}\\
 &= \E_{\theta Y}[\Dkl( \pi_{Y'|\theta Y} || \widetilde\pi_{Y'})]  - \E_{Y}[\Dkl( \pi_{Y'|Y} || \widetilde\pi_{Y'} )].
\end{align*}
 Applying Lemma \ref{lem:dLSI_Fisher} with $X=Y'$ and with either $Z=(\theta,Y)$ or $Z=Y$ yields
 \begin{align*}
 \E_Y[\Dkl( \pi_{\theta|Y,Y'} || \pi_{\theta|Y} )]
 &\leq \frac{1}{2}   \ln\left( \left| \Cov(Y') \E[\mathcal{I}_{Y'|\theta,Y}] \right| \right)  - \frac{1}{2}   \ln\left(\frac{|\Cov(Y')|}{|\E[\Cov(Y'| Y)]|}\right) \\
 \E_Y[\Dkl( \pi_{\theta|Y,Y'} || \pi_{\theta|Y} )]&\geq \frac{1}{2}    \ln\left(\frac{|\Cov(Y')|}{|\E[\Cov(Y'| \theta,Y)]|}\right) -  \frac{1}{2}  \ln\left( \left| \Cov(Y') \E[\mathcal{I}_{Y'|Y}] \right| \right) ,
\end{align*}
which gives the result.
\end{proof}

\section{Proof of Theorem \ref{th:Bound_EIG}}\label{proof:Bound_EIG_incremental}
Let $W_m\in\R^{p\times m}$ and $W_\mathrm{new}\in\R^{p\times m'}$ be two matrices such that $[W_m,W_\mathrm{new}]\in\R^{p\times (m+m')}$ has rank $m+m'\leq p$.
We denote by $W_\perp\in\R^{p\times (p-m-m')}$ any matrix such that the block matrix $W=[W_\mathrm{new},W_m,W_\perp]\in\R^{p\times p}$ is invertible.
Consider the matrix $V=W^{-\top}$ and its block decomposition $V=[V_\mathrm{new},V_m,V_\perp]$.
By construction, the random variables $Y_\mathrm{new}=W_\mathrm{new}^\top Y$, $Y_m=W_m^\top Y$ and $Y_\perp = W_\perp^\top Y$ are such that
\begin{equation}\label{tmp:20756187}
 \begin{pmatrix}
  Y_\mathrm{new} \\ Y_m \\ Y_\perp
 \end{pmatrix}
 = W^\top Y
 \qquad\Leftrightarrow\qquad
 Y= V_\mathrm{new} Y_\mathrm{new} + V_m Y_m +  V_\perp Y_\perp .
\end{equation}
The following Lemma will be useful.
\begin{lemma}\label{lem:FisherMarginal}
With the above notations we have
  \begin{align}
  \E[\mathcal{I}_{Y_\mathrm{new}|Y_m}] &\preceq \Big(  W_\mathrm{new}^\top \Sigma_\mathrm{signal}[W_m] W_\mathrm{new}\Big)^{-1}  \label{tmp:104763} \\
  \E[\mathcal{I}_{Y_\mathrm{new}|\theta,Y_m}] &\preceq \Big(  W_\mathrm{new}^\top \Sigma_\mathrm{noise}[W_m] W_\mathrm{new}\Big)^{-1} \label{tmp:103571} .
 \end{align}
 where we recall the notations
 \begin{align*}
  \Sigma_\mathrm{signal}[W_m] &= \Sigma_\mathrm{signal} - \Sigma_\mathrm{signal} W_m ( W_m^\top \Sigma_\mathrm{signal} W_m )^{-1} W_m^\top \Sigma_\mathrm{signal} \\
  \Sigma_\mathrm{noise}[W_m] &= \Sigma_\mathrm{noise} - \Sigma_\mathrm{noise} W_m ( W_m^\top \Sigma_\mathrm{noise} W_m )^{-1} W_m^\top \Sigma_\mathrm{noise} ,
 \end{align*}
 where $\Sigma_\mathrm{signal}$ and $\Sigma_\mathrm{noise}$ are any matrix such that $\mathcal{I}_Y\preceq\Sigma_\mathrm{signal}^{-1}$ and $\E[\mathcal{I}_{Y|\theta}]\preceq\Sigma_\mathrm{noise}^{-1}$.
\end{lemma}

\begin{proof}
Consider the Fisher information matrix
$$
 \mathcal{I}_{(Y_\mathrm{new},Y_m)} =
 \begin{pmatrix}
   \E[\mathcal{I}_{Y_\mathrm{new}|Y_m}] & \mathcal{I}_{ Y_\mathrm{new} ,Y_m} \\
  \mathcal{I}_{ Y_m, Y_\mathrm{new}} & \E[\mathcal{I}_{Y_m|Y_\mathrm{new}}] \\
 \end{pmatrix}.
$$
Proposition \ref{prop:Fisher} permits to write
$$
 \mathcal{I}_{(Y_\mathrm{new},Y_m)}
 \overset{\eqref{eq:FisherMarginal_bound}}{\preceq}
 S:=
 \E[\mathcal{I}_{(Y_\mathrm{new},Y_m)|Y_\perp}] -
 \mathcal{I}_{(Y_\mathrm{new},Y_m),Y_\perp}
 \E[\mathcal{I}_{Y_\perp|(Y_m,Y_\perp)}]^{-1}
 \mathcal{I}_{Y_\perp, (Y_\mathrm{new},Y_m)} .
$$
Next, we show that $S$ is the Shur complement of a block decomposition of the Fisher information matrix $\mathcal{I}_{Y}$.
By \eqref{tmp:20756187} we have
\begin{align*}
 \nabla_{Y_\mathrm{new}}\ln\pi_{Y_\mathrm{new} , Y_m,Y_\perp}(Y_\mathrm{new} , Y_m,Y_\perp)
 &= V_\mathrm{new}^\top \nabla_Y\ln\pi( Y  ) \\
 \nabla_{Y_m}\ln\pi_{Y_\mathrm{new} , Y_m,Y_\perp}(Y_\mathrm{new} , Y_m,Y_\perp)
 &= V_m^\top \nabla_Y\ln\pi( Y  ) \\
 \nabla_{Y_\perp}\ln\pi_{Y_\mathrm{new} , Y_m,Y_\perp}(Y_\mathrm{new} , Y_m,Y_\perp)
 &= V_\perp^\top \nabla_Y\ln\pi( Y  ) ,
\end{align*}
so that
\begin{align*}
  \E[\mathcal{I}_{Y_\mathrm{new},Y_m|Y_\perp}] &= (V_\mathrm{new},V_m)^\top \mathcal{I}_{Y} (V_\mathrm{new},V_m) \\
  \mathcal{I}_{(Y_\mathrm{new},Y_m),Y_\perp} &=  (V_\mathrm{new},V_m)^\top \mathcal{I}_{Y} V_\perp \\
  \E[\mathcal{I}_{Y_\perp|(Y_m,Y_\perp)}] &=  V_\perp^\top \mathcal{I}_{Y} V_\perp .
\end{align*}
Thus, $S$ is the Shur complement of the block matrix
$$
     V^\top \mathcal{I}_Y V =\begin{pmatrix}
      V_\mathrm{new}^\top \mathcal{I}_Y V_\mathrm{new} & V_\mathrm{new}^\top \mathcal{I}_Y V_m &  V_\mathrm{new}^\top \mathcal{I}_Y V_\perp \\
      V_m^\top \mathcal{I}_Y V_\mathrm{new} & V_m^\top \mathcal{I}_Y V_m & V_m^\top \mathcal{I}_Y V_\perp \\
      V_\perp^\top \mathcal{I}_Y V_\mathrm{new} & V_\perp^\top \mathcal{I}_Y V_m & V_\perp^\top \mathcal{I}_Y V_\perp
     \end{pmatrix} ,
$$
so that
\begin{align*}
 \begin{pmatrix}
\Big(~~S^{-1}~\Big) & \begin{matrix}\times \\ \times \end{matrix} \\
\begin{matrix}\times & \times \end{matrix}  & \begin{matrix}\times  \end{matrix}
\end{pmatrix}
   &=\Big( V^\top \mathcal{I}_Y V \Big)^{-1}\\
   &= W^\top \mathcal{I}_Y^{-1} W \\
   &\preceq W^\top \Sigma_\mathrm{signal} W
   = \begin{pmatrix}
    W_\mathrm{new}\Sigma_\mathrm{signal} W_\mathrm{new} & W_\mathrm{new}\Sigma_\mathrm{signal} W_m & \times \\
     W_m\Sigma_\mathrm{signal} W_\mathrm{new} & W_m\Sigma_\mathrm{signal} W_m & \times\\
     \times& \times& \times
   \end{pmatrix}
\end{align*}
Finally we obtain
\begin{align*}
 \begin{pmatrix}
   \E[\mathcal{I}_{Y_\mathrm{new}|Y_m}] & \mathcal{I}_{ Y_\mathrm{new} ,Y_m} \\
  \mathcal{I}_{ Y_m, Y_\mathrm{new}} & \E[\mathcal{I}_{Y_m|Y_\mathrm{new}}] \\
 \end{pmatrix}
 &= \mathcal{I}_{(Y_\mathrm{new},Y_m)} \\
 &\preceq S \\
 &=\begin{pmatrix}
    W_\mathrm{new}\Sigma_\mathrm{signal} W_\mathrm{new} & W_\mathrm{new}\Sigma_\mathrm{signal} W_m  \\
     W_m\Sigma_\mathrm{signal} W_\mathrm{new} & W_m\Sigma_\mathrm{signal} W_m \\
   \end{pmatrix}^{-1}  \\
 &= \begin{pmatrix}
    \Big( W_\mathrm{new} \Sigma_\mathrm{signal}[W_m] W_\mathrm{new} \Big)^{-1} & \times   \\
     \times & \times \\
   \end{pmatrix}.
\end{align*}
where $\Sigma_\mathrm{signal}[W_m] =
\Sigma_\mathrm{signal}
- \Sigma_\mathrm{signal} W_m ( W_m^\top \Sigma_\mathrm{signal} W_m )^{-1}
W_m^\top \Sigma_\mathrm{signal}$. By considering the first block of the above decomposition, we obtain \eqref{tmp:104763}.
Relation \eqref{tmp:103571} is derived analogously, which completes the proof.
\end{proof}

\paragraph{Proof of \eqref{eq:Bound_EIG}}
We start by decomposing the EIG as follow
\begin{align}
 \mathrm{EIG}(W_m)
 &= \E_{Y} \left[\Dkl( \pi_{\theta|Y_m} || \pi_\theta ) \right] \nonumber\\
 &= \int  \log\frac{\pi_{\theta|Y_m}}{\pi_{\theta}} \d\pi_{\theta Y} \nonumber\\
 &= \int  \log\frac{\pi_{\theta|Y_m,Y_\mathrm{new}}}{\pi_{\theta}} \d\pi_{\theta Y} - \int  \log\frac{\pi_{\theta|Y_m,Y_\mathrm{new}}}{\pi_{\theta|Y_m}} \d\pi_{\theta Y} \nonumber\\
 &= \mathrm{EIG}([W_m,W_\mathrm{new}]) - \E_Y[\Dkl( \pi_{\theta|Y_m,Y_\mathrm{new}} || \pi_{\theta|Y_m} )] .
 \label{tmp:078}
\end{align}
Applying Proposition \ref{prop:TheResult} and Lemma \ref{lem:FisherMarginal} yields
\begin{align*}
 &\mathrm{EIG}(W_m) \\
 &\overset{\eqref{tmp:078}}{=}\mathrm{EIG}([W_m,W_\mathrm{new}]) - \E_Y[\Dkl( \pi_{\theta|Y_m,Y_\mathrm{new}} || \pi_{\theta|Y_m} )] \\
 &\overset{\eqref{eq:TheResult}}{\leq}  \mathrm{EIG}([W_m,W_\mathrm{new}]) - \frac{1}{2}    \ln\left(\frac{| \E[\mathcal{I}_{Y_\mathrm{new} | Y_m}]^{-1} |}{|\E[\Cov(Y_\mathrm{new}| \theta,Y_m)]|}\right) \\
 &\overset{\eqref{tmp:104763}}{\leq}
 \mathrm{EIG}([W_m,W_\mathrm{new}]) - \frac{1}{2}    \ln\left(| W_\mathrm{new}^\top \Sigma_\mathrm{signal}[W_m] W_\mathrm{new}  |  \right) + \frac{1}{2}    \ln\left( |\E[\Cov(Y_\mathrm{new}| \theta,Y_m)]|\right).
\end{align*}
To show \eqref{eq:Bound_EIG}, it remains to bound $\E[\Cov(Y_\mathrm{new}| \theta,Y_m)]$ by $W_\mathrm{new}^\top  \Sigma_{Y|\theta}[W_m]  W_\mathrm{new}$, where $\Sigma_{Y|\theta}[W_m] = \Sigma_{Y|\theta}  - \Sigma_{Y|\theta} W_m (W_m^\top \Sigma_{Y|\theta} W_m)^{-1}  W_m^\top \Sigma_{Y|\theta}$. For any $A\in\R^{m'\times m}$ we can write
\begin{align*}
 \E[\Cov(Y_\mathrm{new}| \theta,Y_m)]
 &=\E[\Cov(Y_\mathrm{new} - A Y_m| \theta,Y_m)] \\
 &\preceq \E[\Cov(Y_\mathrm{new} - A Y_m| \theta)] \\
 &= \E[\Cov( (W_\mathrm{new}^\top - A  W_m^\top) Y| \theta)] \\
 &= (W_\mathrm{new}^\top - A  W_m^\top)\E[\Cov(Y| \theta)] (W_\mathrm{new} - W_m A^\top) .
\end{align*}
By taking $A= W_\mathrm{new}^\top\E[\Cov(Y|\theta)] W_m ( W_m^\top \E[\Cov(Y| \theta)] W_m )^{-1}$ we obtain
$
\E[\Cov(Y_\mathrm{new}| \theta,Y_m)] \preceq
W_\mathrm{new}^\top  \Sigma_{Y|\theta}[W_m]  W_\mathrm{new} .
$
This completes the proof of \eqref{eq:Bound_EIG}.

\paragraph{Proof of \eqref{eq:Bound_EIG_up}}
Applying Proposition \ref{prop:TheResult} and Lemma \ref{lem:FisherMarginal} yields
\begin{align*}
 &\mathrm{EIG}(W_m) \\
 &\overset{\eqref{tmp:078}}{=}\mathrm{EIG}([W_m,W_\mathrm{new}]) - \E_Y[\Dkl( \pi_{\theta|Y_m,Y_\mathrm{new}} || \pi_{\theta|Y_m} )] \\
 &\overset{\eqref{eq:TheResult}}{\geq}  \mathrm{EIG}([W_m,W_\mathrm{new}]) - \frac{1}{2}   \ln\left( \frac{| \E[\Cov(Y_\mathrm{new}| Y_m)] |}{|\E[\mathcal{I}_{Y_\mathrm{new} | \theta,Y_m}]^{-1}|} \right) \\
 &\overset{\eqref{tmp:103571}}{\geq}  \mathrm{EIG}([W_m,W_\mathrm{new}]) + \frac{1}{2}   \ln\left( |W_\mathrm{new}^\top \Sigma_\mathrm{noise}[W_m] W_\mathrm{new}| \right) - \frac{1}{2}   \ln\left( |   \E[\Cov(Y_\mathrm{new}| Y_m)]   |  \right)
\end{align*}
By the same calculation as above, we have
$
 \E[\Cov(Y_\mathrm{new}| Y_m)] \preceq  W_\mathrm{new}^\top  \Sigma_Y[W_m] W_\mathrm{new} ,
$
where $\Sigma_Y[W_m] = \Sigma_{Y}  - \Sigma_{Y} W_m (W_m^\top \Sigma_{Y} W_m)^{-1}  W_m^\top \Sigma_{Y}$. We deduce \eqref{eq:Bound_EIG_up}, which concludes the proof.

\section{Proof of Corollary \ref{cor:Bound_EIG_conservative}}\label{proof:Bound_EIG_conservative}

We prove Corollary \ref{cor:Bound_EIG_conservative} using Theorem \ref{th:Bound_EIG}. For any $W_m$, we choose $W_\text{new}=W_\perp$ be such that $W=[W_m,W_\perp]\in\R^{p\times p}$ is invertible and
\begin{equation}\label{tmp:109586}
 W^\top \Sigma_\text{noise} W =
 \begin{pmatrix}
  W_m^\top \Sigma_\text{noise} W_m & 0 \\ 0 & I_{p-m}
 \end{pmatrix} .
\end{equation}
Such a matrix $W_\perp$ exists as the result of the Gram-Schmidt process.
In particular, since $\Sigma_\text{noise}[W_m] = \Sigma_\text{noise} - \Sigma_\text{noise} W_m ( W_m^\top \Sigma_\text{noise} W_m )^{-1} W_m^\top \Sigma_\text{noise}$ we have
\begin{align}
 W_\perp^\top \Sigma_\text{noise}[W_m] W_\perp
 &\overset{\eqref{tmp:109586}}{=} W_\perp^\top \Sigma_\text{noise} W_\perp .  \label{tmp:07856}
\end{align}
Furthermore, since $\Sigma_Y[W_m] = \Sigma_{Y}  - \Sigma_{Y} W_m (W_m^\top \Sigma_{Y} W_m)^{-1}  W_m^\top \Sigma_{Y}$ we have
\begin{align}
 \ln|W^\top \Sigma_Y W|
 &=
 \ln\left|\begin{pmatrix}
  W_m^\top \Sigma_Y W_m & W_m^\top \Sigma_Y W_\perp \\ W_\perp^\top \Sigma_Y W_m & W_\perp^\top \Sigma_Y W_\perp
 \end{pmatrix}\right| \nonumber\\
 &= \ln|W_m^\top \Sigma_Y W_m| + \ln|W_\perp \Sigma_Y[W_m] W_\perp| .  \label{tmp:13057}
\end{align}
Then, Inequality \eqref{eq:Bound_EIG_up} writes
\begin{align*}
 \mathrm{EIG}(I_d)
 &= \mathrm{EIG}(W) \\
 &\overset{\eqref{eq:Bound_EIG_up}}{\leq} \mathrm{EIG}(W_m) + \frac{1}{2}   \ln\left( \frac{ | W_\perp^\top \Sigma_Y[W_m]W_\perp | }{ |W_\perp^\top \Sigma_\text{noise}[W_m] W_\perp| }\right) \\
 &\overset{\eqref{tmp:07856}}{=} \mathrm{EIG}(W_m) + \frac{1}{2}   \ln\left( | W_\perp^\top \Sigma_Y[W_m]W_\perp | \right) \\
 &\overset{\eqref{tmp:13057}}{=}¨ \mathrm{EIG}(W_m) + \frac{1}{2}   \ln\left( \frac{ | W^\top \Sigma_Y W | }{ |W_m^\top \Sigma_Y W_m|}\right) \\
 &= \mathrm{EIG}(W_m) + \frac{1}{2}   \ln\left( \frac{ | W^\top \Sigma_\text{noise} W | }{ |W_m^\top \Sigma_Y W_m|}\right) -  \frac{1}{2}   \ln\left( \frac{ | \Sigma_\text{noise} |}{ | \Sigma_Y | }\right)\\
 &\overset{\eqref{tmp:109586}}{=} \mathrm{EIG}(W_m) + \frac{1}{2}   \ln\left( \frac{ | W_m^\top \Sigma_\text{noise} W_m | }{ |W_m^\top \Sigma_Y W_m|}\right) -  \frac{1}{2}   \ln\left( \frac{ | \Sigma_\text{noise} |}{ | \Sigma_Y | }\right) .
\end{align*}
which is implies \eqref{eq:Bound_EIG_conservative}.
Proceeding the same way, we have that \eqref{eq:Bound_EIG} implies \eqref{eq:Bound_EIG_conservative_up}. This concludes the proof.

\section{Proof of Proposition \ref{prop:quasi_opt_cons}}\label{proof:quasi_opt_cons}

 We first note that for any three real-valued functions $\underline{f},f,\overline{f}$ defined over a set $\mathcal{X}$ such that, for all $x\in\mathcal{X}$, we have
 $$
  \underline{f}(x) \overset{(a)}{\leq} f(x) \overset{(b)}{\leq} \overline{f}(x) ,
 $$
 then it holds
 \begin{equation}\label{eq:temp_f_bound}
  f( \underline{x} ) \geq \sup_{x\in\mathcal{X}} f(x) - \sup_{x\in\mathcal{X}}( \overline{f}(x)-\underline{f}(x) ) ,
 \end{equation}
 where $\underline{x}\in\argmax_{x\in\mathcal{X}} \underline{f}(x)$, assuming it exists.
 Indeed, for any $x\in\mathcal{X}$ we have
 \begin{align*}
  f( \underline{x} ) \overset{(a)}{\geq}  \underline{f}(\underline{x})
  \geq \underline{f}(x)
  \overset{(b)}{\geq}f(x) -  ( \overline{f}(x)-\underline{f}(x) )
  \geq f(x) -  \sup_{x\in\mathcal{X}}( \overline{f}(x)-\underline{f}(x) ) ,
 \end{align*}
 which yields \eqref{eq:temp_f_bound}.
 \textbf{Incremental bounds:}
 applying \eqref{eq:temp_f_bound} on the inequalities \eqref{eq:Bound_EIG_incremental}--\eqref{eq:Bound_EIG_incremental_up} yields
 \begin{align*}
  \mathrm{EIG}(W_m^{\text{inc}})
  &\geq
  \sup_{W_m \in \mathcal{K}_m} \mathrm{EIG}(W_m)
  - \sup_{W_m \in \mathcal{K}_m} \frac{1}{2}   \ln\left( \frac{ | W_m^\top \Sigma_{Y} W_m | }{ |W_m^\top \Sigma_\text{noise} W_m| } \frac{| W_m^\top \Sigma_{Y|\theta}  W_m |}{|W_m^\top \Sigma_\text{signal}  W_m|}\right) \\
  &\geq \sup_{W_m \in \mathcal{K}_m} \mathrm{EIG}(W_m)
  - \frac{1}{2}\sup_{W_m \in \R^{p\times m}}    \ln\left( \frac{ | W_m^\top \Sigma_{Y} W_m | }{|W_m^\top \Sigma_\text{signal}  W_m|}\right) \\
  &-  \frac{1}{2}\sup_{W_m \in \R^{p\times m}}    \ln\left( \frac{| W_m^\top \Sigma_{Y|\theta}  W_m |}{ |W_m^\top \Sigma_\text{noise} W_m| }\right) \\
  &= \sup_{W_m \in \mathcal{K}_m} \mathrm{EIG}(W_m) - \sum_{i=1}^m\frac{\ln(\alpha_i) + \ln(\beta_i)}{2}
 \end{align*}
 which is \eqref{eq:quasi_opt_inc}.
 The \textbf{conservative bounds} are obtained in the same way: applying \eqref{eq:temp_f_bound} on \eqref{eq:Bound_EIG_conservative}--\eqref{eq:Bound_EIG_conservative_up} we get
 \begin{align*}
  &\mathrm{EIG}(W_m^{\text{cons}})\\
  &\geq
  \sup_{W_m \in \mathcal{K}_m} \mathrm{EIG}(W_m)
  - \sup_{W_m \in \mathcal{K}_m} \frac{1}{2}   \ln\left( \frac{|W_m^\top \Sigma_\text{signal}  W_m|}{| W_m^\top \Sigma_{Y|\theta}  W_m |}\frac{ |W_m^\top \Sigma_\text{noise} W_m| }{ | W_m^\top \Sigma_{Y} W_m | } \right) \\
  &+ \frac{1}{2}\ln\left( \frac{|\Sigma_\text{signal}  |}{| \Sigma_Y |}  \frac{|\Sigma_\text{noise}| }{|\Sigma_{Y|\theta}|}\right) \\
  &\geq
  \sup_{W_m \in \mathcal{K}_m} \mathrm{EIG}(W_m)
  - \sup_{W_m \in \R^{p\times m}}  \frac{1}{2}   \ln\left( \frac{ |W_m^\top \Sigma_\text{noise} W_m| }{| W_m^\top \Sigma_{Y|\theta}  W_m |} \right)  \\
  &- \sup_{W_m \in \R^{p\times m}}  \frac{1}{2}   \ln\left( \frac{|W_m^\top \Sigma_\text{signal}  W_m|}{ | W_m^\top \Sigma_{Y} W_m | } \right)
  + \frac{1}{2}\ln\left( \frac{|\Sigma_\text{signal}  |}{| \Sigma_Y |}  \frac{|\Sigma_\text{noise}| }{|\Sigma_{Y|\theta}|}\right) \\
  &= \sup_{W_m \in \mathcal{K}_m} \mathrm{EIG}(W_m)
  - \sum_{i=d-m+1}^d \frac{\ln(\alpha_i^{-1})+\ln(\beta_i^{-1})}{2}  + \sum_{i=1}^d \frac{\ln(\alpha_i^{-1})+\ln(\beta_i^{-1})}{2} \\
  &= \sup_{W_m \in \mathcal{K}_m} \mathrm{EIG}(W_m)
  - \sum_{i=1}^{d-m} \frac{\ln(\alpha_i)+\ln(\beta_i)}{2}
 \end{align*}
 which is \eqref{eq:quasi_opt_cons}.

\section{Proof of Proposition \ref{prop:ComputeFisher}}\label{proof:ComputeFisher}

The fact that $\Sigma_Y:=\Cov(Y) = \Sigma_\text{obs} + \Cov( G(\eta) )$ and $\Sigma_{Y|\theta}:=\E[\Cov(Y|\theta)] = \Sigma_\text{obs} + \E[ \Cov( G(\eta) |\theta) ] $ are direct consequences of the independence between $\eta$ and $\varepsilon$.
To compute the Fisher information $\E[\mathcal{I}_{Y|\theta}]$, we use the fact that $Y|\eta \sim\mathcal{N}( G(\eta) ,\Sigma_\text{obs})$ is Gaussian with $\pi_{Y|\eta}(y|\eta)\propto \exp(-\frac{1}{2}\| G(\eta) - y \|^2_{\Sigma_\text{obs}^{-1}}) ,$
and then $ \nabla_Y \ln\pi_{Y|\eta}(Y|\eta)  =  \Sigma_\text{obs}^{-1}( G(\theta,\eta)-Y )$.
Thus, Proposition \ref{prop:Fisher} permits to write
\begin{align*}
 \E[\mathcal{I}_{Y|\theta}]
 &\overset{\eqref{eq:FisherMarginal_bis}}{=} \E[\mathcal{I}_{Y|\theta,\eta}] - \E[\Cov(\nabla_Y \ln\pi_{\theta,\eta,Y}(\theta,\eta,Y)|Y,\theta)] \\
 &=\E[\mathcal{I}_{Y|\theta,\eta}] - \E[\Cov(\nabla_Y \ln\pi_{Y|\theta,\eta}(Y|\theta,\eta)|Y,\theta)] \\
 &=\E[\mathcal{I}_{Y|\eta}] - \E[\Cov(\nabla_Y \ln\pi_{Y|\eta}(Y|\eta)|Y,\theta)] \\
 &= \Sigma_\text{obs}^{-1} -  \E[\Cov(  \Sigma_\text{obs}^{-1}( G(\theta,\eta)-Y ) |Y,\theta)] \\
 &= \Sigma_\text{obs}^{-1} -  \Sigma_\text{obs}^{-1}\E[\Cov(   G(\theta,\eta) |Y,\theta)] \Sigma_\text{obs}^{-1} ,
\end{align*}
where we used the fact that $(Y|\eta,\theta) = (Y|\eta)$ and that $\mathcal{I}_{Y|\eta}=\Sigma_\text{obs}^{-1}$.
To compute $\mathcal{I}_Y$, let us write
\begin{align*}
 \nabla_Y \ln\pi_{\theta,Y}(\theta,Y)
 &= \nabla_Y \ln\pi_{Y|\theta}(Y|\theta) \\
 &= \frac{\nabla_Y \int \pi_{Y,\eta|\theta}(Y,\eta|\theta) \d\eta}{\pi_{Y|\theta}(Y|\theta)} \\
 &= \frac{ \int \nabla_Y\ln\pi_{Y,\eta|\theta}(Y,\eta|\theta)  \pi_{Y,\eta|\theta}(Y,\eta|\theta) \d\eta}{\pi_{Y|\theta}(Y|\theta)} \\
 &=  \int \nabla_Y\ln\pi_{Y,\eta|\theta}(Y,\eta|\theta)  \pi_{\eta|Y,\theta}(\eta|Y,\theta) \d\eta \\
 &= \E[ \nabla_Y\ln\pi_{Y|\eta}(Y|\eta) |\theta,Y] ,
\end{align*}
where, in the last step, we used the fact that $(Y|\theta,\eta)=(Y|\eta)$. Then, using the law of total variance $\Cov(Z|Y) = \E[\Cov(Z|T)|Y] +\Cov(\E[Z|T]|Y)$ with $Z=\nabla_Y\ln\pi_{Y|\eta}(Y|\eta)$ and $T=(\theta,Y)$, we can write
\begin{align*}
 \mathcal{I}_{Y}
 &\overset{\eqref{eq:FisherMarginal}}{=} \E[\mathcal{I}_{Y|\theta}] - \E[\Cov(\nabla_Y \ln\pi_{\theta,Y}(\theta,Y)|Y)] \\
 &= \E[\mathcal{I}_{Y|\theta}] - \E\Big[\Cov( \E[ \nabla_Y\ln\pi_{Y|\eta}(Y|\eta) |\theta,Y]  |Y) \Big] \\
 &= \E[\mathcal{I}_{Y|\theta}] - \E\Big[  \Cov(\nabla_Y\ln\pi_{Y|\eta}(Y|\eta)|Y) - \E[\Cov(\nabla_Y\ln\pi_{Y|\eta}(Y|\eta)|\theta,Y)|Y] \Big] \\
 &= \underbrace{\E[\mathcal{I}_{Y|\theta}] +  \E[\Cov(\nabla_Y\ln\pi_{Y|\eta}(Y|\eta)|\theta,Y)] }_{=\Sigma_\text{obs}^{-1}} - \underbrace{\E[  \Cov(\nabla_Y\ln\pi_{Y|\eta}(Y|\eta)|Y) ]}_{=\Sigma_\text{obs}^{-1}\E[\Cov(   G(\theta,\eta) |Y)] \Sigma_\text{obs}^{-1}},
\end{align*}
where, in the last step, we used the above expression of $\E[\mathcal{I}_{Y|\theta}]$. This concludes the proof.

\section{Proof of Proposition \ref{prop:Sigma_FG}}\label{proof:Sigma_FG}

 We start by showing that $\Sigma_\text{signal}^{(\infty,\mathcal{F})}$ is well defined.
 Because $((1+\alpha) f+\beta\ell)\in\mathcal{F}$ for any $\alpha,\beta\in\R$ and any affine function $\ell$, we have
 \begin{align*}
  \E[\|G(\eta)-f(Y)\|^2]
  &\overset{\eqref{eq:conditionalExpectation_Y}}{\leq} \E[\|G(\eta)-(1+\alpha)f(Y) - \beta\ell(Y)\|^2] \\
  &= \E[\|G(\eta)-f(Y)\|^2] - 2\alpha \E[ (G(\eta)-f(Y))^\top f(Y)\|^2]  \\
  &- 2 \beta \E[ (G(\eta)-f(Y))^\top \ell(Y)\|^2]  +\mathcal{O}(\alpha^2 + \alpha\beta+ \beta^2 ) ,
 \end{align*}
 so that $\E[ (G(\eta)-f(Y))^\top f(Y)\|^2]=0$ and $\E[ (G(\eta)-f(Y))^\top \ell(Y)\|^2]=0$. We deduce
\begin{align*}
 &\E[(G(\eta) - f(Y) )^{\otimes 2}] \\
 &= \E[(G(\eta) - \ell(Y) )^{\otimes 2}] +2 \E[ (G(\eta) - \ell(Y) )(\ell(Y)  - f(Y))^\top ]_\text{sym}+ \E[(\ell(Y)  - f(Y))^{\otimes 2}] \\
 &= \E[(G(\eta) - \ell(Y) )^{\otimes 2}] +2 \E[ (f(Y) - \ell(Y) )(\ell(Y)  - f(Y))^\top ]_\text{sym}+ \E[(\ell(Y)  - f(Y))^{\otimes 2}] \\
 &= \E[(G(\eta) - \ell(Y) )^{\otimes 2}] - \E[(\ell(Y)  - f(Y))^{\otimes 2}] \\
 &\preceq \E[(G(\eta) - \ell(Y) )^{\otimes 2}].
\end{align*}
Next we take $\ell(Y)=A(Y-\E[G(\eta)])-\E[G(\eta)]$ so that
\begin{align*}
 \E[(G(\eta) - \ell(Y) )^{\otimes 2}]
 &= \E[(  G(\eta)-\E[G(\eta)] - A( G(\eta)-\E[G(\eta)] - \varepsilon ) )^{\otimes 2}] \\
 &= (I-A)\Cov(G(\eta))(I-A)^\top + A \Sigma_\text{obs}A^\top \\
 &= A(\Cov(G(\eta))+\Sigma_\text{obs})A^\top - 2 (A\Cov(G(\eta)))_\text{sym} + \Cov(G(\eta)) .
\end{align*}
Choosing $A = \Cov(G(\eta)) (\Cov(G(\eta))+\Sigma_\text{obs})^{-1}$ yields
\begin{align*}
 \E[(G(\eta) - \ell(Y) )^{\otimes 2}]
 &= \Cov(G(\eta)) A^\top - 2 (A\Cov(G(\eta)))_\text{sym} + \Cov(G(\eta)) \\
 &=  -  A\Cov(G(\eta))  + \Cov(G(\eta)) \\
 &=  \Cov(G(\eta))  -  \Cov(G(\eta)) (\Cov(G(\eta))+\Sigma_\text{obs})^{-1}\Cov(G(\eta))  \\
 &= (\Cov(G(\eta))^{-1}+\Sigma_\text{obs}^{-1})^{-1} \\
 &=  \Sigma_\text{obs}  - \Sigma_\text{obs} ( \Sigma_\text{obs} +\Cov(G(\eta)))^{-1}\Sigma_\text{obs},
\end{align*}
where in the last steps we used the Woodbury formula. Since $\Sigma_\text{obs} \succ0 $ and $\Cov(G(\eta)) \succ0$ we deduce that $\E[(G(\eta) - \ell(Y) )^{\otimes 2}]\prec \Sigma_\text{obs}$ so that $\Sigma_\text{signal}^{(\infty,\mathcal{F})}$ is well defined.

We show now that $\Sigma_\text{noise}^{(\infty,\mathcal{G})}$ is well defined. Using the same arguments, we can show that $ \E[(G(\eta) - g(Y,\theta) )^{\otimes 2}]  \preceq \E[(G(\eta) - \ell(Y,\theta) )^{\otimes 2}]$ for any affine function $\ell(Y,\theta)$. In particular with the same $\ell(Y,\theta)=\ell(Y)$ as above, we get that $\E[(G(\eta) - \ell(Y,\theta) )^{\otimes 2}]\prec \Sigma_\text{obs}$, therefore $\Sigma_\text{noise}^{(\infty,\mathcal{G})}$ is well defined.

\section{Proof of Proposition \ref{prop:BoundFisher}}\label{proof:BoundFisher}

The proof consists in applying Proposition \ref{prop:Fisher}.
\paragraph{Proof of \eqref{eq:SigmaSignal_boundFisher}.} We bound $\mathcal{I}_Y$ using
\begin{align*}
 \mathcal{I}_Y
  &\overset{\eqref{eq:FisherMarginal_bound}}{\preceq}   \E[\mathcal{I}_{Y|\eta}]  - \mathcal{I}_{Y, \eta } \E[\mathcal{I}_{\eta|Y}]^{-1}  \mathcal{I}_{\eta,Y} .
\end{align*}
In order to compute the terms above, let us recall $\pi_{Y,\eta}(Y,\eta) \propto \exp(-\frac{1}{2}\|G(\eta)-Y\|^2_{\Sigma_\text{obs}^{-1}})\pi_{\eta}(\eta)$ so that
\begin{align*}
  \nabla_Y\ln\pi_{Y,\eta}(Y,\eta) &=  \Sigma_\text{obs}^{-1}(G(\eta)-Y) \\
  \nabla_\eta\ln\pi_{Y,\eta}(Y,\eta) &=  -  \Jac  G(\eta)^\top \Sigma_\text{obs}^{-1} ( G(\eta)-Y ) + \nabla_\eta\ln\pi_{\eta}(\eta)
\end{align*}
Because $G(\eta)-Y = -\varepsilon$ we deduce
$$
 \E[\mathcal{I}_{Y|\eta}]
 = \Cov(\Sigma_\text{obs}^{-1}\varepsilon) = \Sigma_\text{obs}^{-1}
$$
and
\begin{align*}
 \mathcal{I}_{Y,\eta }
 &= \E[( \nabla_Y\ln\pi_{Y,\eta}(Y,\eta) )(\nabla_\eta\ln\pi_{Y,\eta}(Y,\eta))^\top] \\
 &= \E[( - \Sigma_\text{obs}^{-1}\varepsilon  )(\Jac  G(\eta)^\top \Sigma_\text{obs}^{-1} \varepsilon + \nabla_\eta\ln\pi_{\eta}(\eta))^\top] \\
 &= -\Sigma_\text{obs}^{-1}\E[ \Jac  G(\eta) ]  ,
\end{align*}
and
\begin{align*}
 \E[\mathcal{I}_{\eta|Y}]
 &= \E[(\nabla_\eta\ln\pi_{Y,\eta}(Y,\eta))^{\otimes 2}] \\
 &= \E[((\Jac  G(\eta)^\top \Sigma_\text{obs}^{-1} \varepsilon + \nabla_\eta\ln\pi_{\eta}(\eta)))^{\otimes 2}] \\
 &= \E[(\Jac  G(\eta)^\top \Sigma_\text{obs}^{-1} \varepsilon))^{\otimes 2}] + \E[(  \nabla_\eta\ln\pi_{\eta}(\eta))^{\otimes 2}] \\
 &= \E[ \Jac  G(\eta)^\top \Sigma_\text{obs}^{-1} \Jac  G(\eta) ] + \mathcal{I}_\eta .
\end{align*}
Combining the above expressions yields
\begin{align*}
 \mathcal{I}_Y
  &\preceq   \Sigma_\text{obs}^{-1}   - \Sigma_\text{obs}^{-1}\E[ \Jac  G(\eta) ]   \Big(  \E[ \Jac  G(\eta)^\top \Sigma_\text{obs}^{-1} \Jac  G(\eta) ] + \mathcal{I}_\eta  \Big)^{-1} \E[ \Jac  G(\eta)^\top ] \Sigma_\text{obs}^{-1}   \\
  &\overset{\eqref{eq:Hu}\&\eqref{eq:Lu}}{=}
  \Sigma_\text{obs}^{-1}   - \Sigma_\text{obs}^{-1} \mathcal{L}(G)   \Big(  \mathcal{H}(G) + \mathcal{L}(G)^\top \Sigma_\text{obs}^{-1} \mathcal{L}(G) + \mathcal{I}_\eta  \Big)^{-1} \mathcal{L}(G)^\top \Sigma_\text{obs}^{-1} \\
  &= \left( \Sigma_\text{obs}   + \mathcal{L}(G)   \Big(  \mathcal{H}(G) + \mathcal{I}_\eta  \Big)^{-1} \mathcal{L}(G)^\top \right)^{-1}
\end{align*}
where we used Woodbury formula $(A + U^\top B^{-1} U)^{-1}=A^{-1} - A^{-1} U^\top (B + UA^{-1} U^\top)^{-1} U A^{-1}$. This shows that $\Sigma_\text{signal}$ defined by \eqref{eq:SigmaSignal_boundFisher} satisfies $\mathcal{I}_Y \preceq \Sigma_\text{signal}^{-1}$.

\paragraph{Proof of \eqref{eq:SigmaNoise_boundFisher}.}
We bound $\E[\mathcal{I}_{Y|\theta}]$ using
\begin{align*}
 \E[\mathcal{I}_{Y|\theta}]
 &\overset{\eqref{eq:FisherMarginal_bound_bis}}{\preceq}
  \E[\mathcal{I}_{Y|\theta,\eta}]  - \E[\mathcal{I}_{Y,\eta|\theta}] \E[\mathcal{I}_{\eta|(Y,\theta)}]^{-1} \E[\mathcal{I}_{\eta,Y|\theta}] .
\end{align*}
In order to compute the terms above, recall that the joint density of $(Y,\eta,\theta)$ is
\begin{equation*}
   \pi_{Y,\eta,\theta}(Y,\eta,\theta) \propto   \exp\left(-\frac{1}{2}\| G(\eta) - Y\|_{\Sigma_\text{obs}^{-1}}^2\right) \exp\left(-\frac{1}{2}\| H(\eta) - \theta\|_{\Sigma_\xi^{-1}}^2 \right) \pi_{\eta}(\eta) ,
\end{equation*}
so that
\begin{align*}
 \nabla_Y \ln\pi_{Y,\eta,\theta}(Y,\eta,\theta) &= \Sigma_\text{obs}^{-1}(G(\eta)-Y) \\
 \nabla_\theta \ln\pi_{Y,\eta,\theta}(Y,\eta,\theta) &=  \Sigma_\xi^{-1}( H(\eta)-\theta) \\
 \nabla_\eta \ln\pi_{Y,\eta,\theta}(Y,\eta,\theta) &=  -  \Jac  G(\eta)^\top \Sigma_\text{obs}^{-1} ( G(\eta)-Y ) \\
 &~~~\, -  \Jac  H(\eta)^\top \Sigma_\xi^{-1} ( H(\eta)-\theta ) + \nabla_\eta\ln\pi_{\eta}(\eta) .
\end{align*}
Because $G(\eta)-Y = -\varepsilon$ and $H(\eta)-\theta = -\xi$ we have
\begin{align*}
 \E[\mathcal{I}_{Y|\theta,\eta}] = \Cov(\Sigma_\text{obs}^{-1}\varepsilon) = \Sigma_\text{obs}^{-1}
\end{align*}
and
\begin{align*}
 \E[\mathcal{I}_{Y,\eta|\theta}]
 &= \E[( \nabla_Y \ln\pi_{Y,\eta,\theta}(Y,\eta,\theta) )( \nabla_\eta \ln\pi_{Y,\eta,\theta}(Y,\eta,\theta) )^\top] \\
 &= \E[( -\Sigma_\text{obs}^{-1}\varepsilon )(  \Jac  G(\eta)^\top \Sigma_\text{obs}^{-1} \varepsilon  +  \Jac  H(\eta)^\top \Sigma_\xi^{-1} \xi + \nabla_\eta\ln\pi_{\eta}(\eta) )^\top] \\
 &= \E[( -\Sigma_\text{obs}^{-1}\varepsilon )(  \Jac  G(\eta)^\top \Sigma_\text{obs}^{-1} \varepsilon   )^\top] \\
 &= -\Sigma_\text{obs}^{-1}\E[ \Jac  G(\eta) ] ,
\end{align*}
and
\begin{align*}
 \E[\mathcal{I}_{\eta|(Y,\theta)}]
 &= \E[( \nabla_\eta \ln\pi_{Y,\eta,\theta}(Y,\eta,\theta) )^{\otimes 2} ] \\
 &= \E[(   \Jac  G(\eta)^\top \Sigma_\text{obs}^{-1} \varepsilon  +  \Jac  H(\eta)^\top \Sigma_\xi^{-1} \xi + \nabla_\eta\ln\pi_{\eta}(\eta) )^{\otimes 2} ] \\
 &= \E[(   \Jac  G(\eta)^\top \Sigma_\text{obs}^{-1} \varepsilon   )^{\otimes 2} ]
 + \E[(   \Jac  H(\eta)^\top \Sigma_\xi^{-1} \xi   )^{\otimes 2} ]
 + \E[(   \nabla_\eta\ln\pi_{\eta}(\eta) )^{\otimes 2} ]\\
 &= \E[ \Jac  G(\eta)^\top \Sigma_\text{obs}^{-1} \Jac  G(\eta) ]
 + \E[  \Jac  H(\eta)^\top \Sigma_\xi^{-1}\Jac  H(\eta) ]
 + \mathcal{I}_{\eta} \\
 &= \mathcal{H}(G) + \mathcal{L}(G)^\top\Sigma_\text{obs}^{-1}\mathcal{L}(G)  + \mathcal{J}(H)  + \mathcal{I}_{\eta}
\end{align*}
We deduce
\begin{align*}
 \E[\mathcal{I}_{Y|\theta}]
 &\preceq  \E[\mathcal{I}_{Y|\theta,\eta}]  - \E[\mathcal{I}_{Y,\eta|\theta}] \E[\mathcal{I}_{\eta|(Y,\theta)}]^{-1} \E[\mathcal{I}_{\eta,Y|\theta}] \\
 &= \Sigma_\text{obs}^{-1}  - \Sigma_\text{obs}^{-1}\mathcal{L}(G)^\top
 \Big(  \mathcal{H}(G) + \mathcal{L}(G)^\top\Sigma_\text{obs}^{-1}\mathcal{L}(G)  + \mathcal{J}(H)  + \mathcal{I}_{\eta} \Big)^{-1}
 \mathcal{L}(G) \Sigma_\text{obs}^{-1} \\
 &=\left( \Sigma_\text{obs}   + \mathcal{L}(G)   \Big(  \mathcal{H}(G) + \mathcal{J}(H) + \mathcal{I}_\eta  \Big)^{-1} \mathcal{L}(G)^\top  \right)^{-1}
\end{align*}
This shows that $\Sigma_\text{noise}$ defined by \eqref{eq:SigmaNoise_boundFisher} satisfies $\E[\mathcal{I}_{Y|\theta}] \preceq \Sigma_\text{noise}^{-1}$ and concludes the proof.

\section{Proof of Equation \eqref{eq:EIG_Gaussian}}\label{proof:EIG_Gaussian}

 By construction, the joint random variable $(\theta,Y)$ is Gaussian with
 $$
 \begin{pmatrix}
  \theta\\Y
 \end{pmatrix}
 \sim \mathcal{N}\left(
 \begin{pmatrix}
  m_\theta \\ A m_\theta
 \end{pmatrix},
 \begin{pmatrix}
  \Sigma_\theta  &   \Sigma_\theta A^\top \\ A\Sigma_\theta  &  \Cov(Y)
 \end{pmatrix}
 \right),
 $$
 where  $\Cov(Y) =  A \Sigma_\theta A^\top  + \Sigma_\text{obs}$.
 We deduce
 $$
 \begin{pmatrix}
  \theta\\Y_m
 \end{pmatrix}
 \sim \mathcal{N}\left(
 \begin{pmatrix}
  m_\theta \\ W_m^\top(  A m_\theta  )
 \end{pmatrix},
 \begin{pmatrix}
  \Sigma_\theta  &   \Sigma_\theta A^\top  W_m\\ W_m^\top A\Sigma_\theta  &  W_m^\top \Cov(Y) W_m
 \end{pmatrix}
 \right),
 $$
 Recall the formula
 $\Dkl(\mu_1||\mu_2) = \frac{1}{2} ( \trace(\Sigma_2^{-1}\Sigma_1)-d+\ln\tfrac{|\Sigma_2|}{|\Sigma_1|} )$
 for any Gaussian measures $\mu_1=\mathcal{N}(m_0,\Sigma_1)$ and $\mu_2=\mathcal{N}(m_0,\Sigma_2)$ with $m_0\in\R^d $ and $ \Sigma_1,\Sigma_2\in\R^{d\times d}$.
 Then
 \begin{align*}
  &\mathrm{EIG}(W_m)\\
  &= \Dkl( \pi_{\theta Y_m} || \pi_\theta\pi_{Y_m} ) \\
  &= \frac{1}{2}\left(\trace \left(\Sigma_{\theta Y_m}
  \begin{pmatrix}
   \Sigma_\theta & 0 \\ 0 & W_m^\top \Cov(Y) W_m
  \end{pmatrix}^{-1} \right)
  -(d+m)+\ln\frac{\left|\begin{pmatrix}
   \Sigma_\theta & 0 \\ 0 & W_m^\top \Cov(Y) W_m
  \end{pmatrix}\right|}{|\Sigma_{\theta Y_m}|}\right) \\
  &= \frac{1}{2}\ln\left|\begin{pmatrix}
   \Sigma_\theta & 0 \\ 0 & W_m^\top \Cov(Y) W_m
  \end{pmatrix}\right|  - \frac{1}{2} \ln \left|
  \begin{pmatrix}
  \Sigma_\theta  &   \Sigma_\theta A^\top  W_m\\ W_m^\top A\Sigma_\theta  &  W_m^\top \Cov(Y) W_m
 \end{pmatrix}
  \right|\\
 &= \frac{1}{2}\ln \left( |\Sigma_\theta|  | W_m^\top \Cov(Y) W_m| \right) \\
 & \qquad - \frac{1}{2} \ln \left( | \Sigma_\theta - \Sigma_\theta A^\top W_m (W_m^\top \Cov(Y) W_m)^{-1} W_m^\top A \Sigma_\theta | | W_m^\top \Cov(Y) W_m| \right) \\
 &= - \frac{1}{2}\ln \left( | I_d - \Sigma_\theta^{1/2} A^\top W_m \left( W_m^\top \Cov(Y) W_m \right)^{-1} W_m^\top A \Sigma_\theta^{1/2} | \right) \\
 &= - \frac{1}{2}\ln \left( | I_d - W_m^\top A  \Sigma_\theta A^\top W_m \left( W_m^\top \Cov(Y) W_m \right)^{-1}  | \right),
 \end{align*}
 where, in the last step, we used the matrix determinant lemma $|I-AB| = | I - BA|$.
 Finally, because $\Cov(Y) = A \Sigma_\theta A^\top  + \Sigma_\text{obs} $, we can write
 \begin{align*}
  \mathrm{EIG}(W_m)
  &= - \frac{1}{2}\ln \left( |  W_m^\top (\Cov(Y) - A  \Sigma_\theta A^\top ) W_m \left( W_m^\top \Cov(Y) W_m \right)^{-1}  | \right) \\
  &= - \frac{1}{2}\ln \left( |  W_m^\top \Sigma_\text{obs}   W_m \left( W_m^\top \Cov(Y) W_m \right)^{-1}  | \right),
 \end{align*}
 which is \eqref{eq:EIG_Gaussian}.

\bibliographystyle{siam}
\bibliography{references}

\end{document}